\pdfoutput=1
 
\documentclass[twocolumn]{article}    

\usepackage{graphicx}          
\usepackage{geometry}
\usepackage{hyperref}
\hypersetup{pdfauthor={Mario Rosenfelder, Henrik Ebel, Jasmin Krauspenhaar, Peter Eberhard},pdftitle={Model Predictive Control of Non-Holonomic Vehicles: Beyond Differential-Drive},colorlinks=false, hidelinks}
\usepackage{amsmath,amssymb,amsthm}
\usepackage{mathtools} 
\usepackage{bm}
\usepackage{xcolor}
\newcommand*{\tran}{^{\mkern-1.5mu\mathsf{T}}} 

\usepackage{pgfplots}  
\pgfplotsset{compat=newest}
\usepackage{tikz, listofitems}
\usetikzlibrary{calc, shapes, backgrounds, bending, decorations.pathmorphing, positioning, decorations.pathreplacing, arrows.meta}
\usepgflibrary{arrows}
\usepackage{subcaption}
\usepackage{enumitem} 
\usepackage{units}
\usepackage{adjustbox}

\newtheorem{lem}{Lemma}
\newtheorem{thm}{Theorem}

\begin{document}

\definecolor{UniBlue}{cmyk}{70,0,0,0}%
\definecolor{grayCar}{RGB}{230, 230, 230}%
\definecolor{colorPlot}{RGB}{44,123,182} 
\definecolor{colorPlot2}{RGB}{253,174,97} 
\def\linewidthPlots{0.4mm}

\tikzset{
	pics/KinCar/.style n args={8}{
		code = { %
        \pgfmathsetmacro\ChassisL{1.5*#4} 
        \pgfmathsetmacro\shiftChassisX{-0.175*\ChassisL} 
        \pgfmathsetmacro\shiftChassisY{-0.5*#5} 
        \pgfmathsetmacro\lengthWheel{0.2*\ChassisL} 
        \pgfmathsetmacro\heightWheel{0.075*#5} 
        \pgfmathsetmacro\wheelxR{-0.5*\lengthWheel} 
        \pgfmathsetmacro\wheelLy{0.2*#5+0.5*\heightWheel} 
        \pgfmathsetmacro\wheelRy{-0.2*#5-0.15*#5} 
        \pgfmathsetmacro\FrontAxisX{#4} 
        \pgfmathsetmacro\wheelxF{\FrontAxisX-0.5*\lengthWheel} 
        \pgfmathsetmacro\wheelLyC{\wheelLy+0.5*\heightWheel} 
        \pgfmathsetmacro\wheelRyC{\wheelRy+0.5*\heightWheel} 
        \pgfmathsetmacro\lengthO{0.5*\ChassisL} 
        \pgfmathsetmacro\linewidthOrient{2*\linewidthPlots} 

        \begin{scope}[shift={(axis direction cs:#1,#2)}, rotate around={#3:(axis cs:0,0)}] 
                \begin{scope}[transparency group,opacity= #8] 
                \draw[draw = black, #7] (axis cs:\shiftChassisX,\shiftChassisY) rectangle ++(axis direction cs:\ChassisL,#5); 
                \draw[very thin] (0,\wheelRy) -- (0,\wheelLy); 
                \draw[very thin] (0,0) -- (axis cs:\FrontAxisX,0); 
                \draw[very thin] (\FrontAxisX,\wheelRy) -- (\FrontAxisX,\wheelLy); 
                \draw[draw = black, fill = black] (axis cs:\wheelxR,\wheelLy) rectangle ++(axis direction cs:\lengthWheel,\heightWheel); 
                \draw[draw = black, fill = black] (axis cs:\wheelxR,\wheelRy) rectangle ++(axis direction cs:\lengthWheel,\heightWheel); 
                \draw[draw = black, fill = black, rotate around={#6:(\FrontAxisX,\wheelLyC)}] (axis cs:\wheelxF,\wheelLy) rectangle ++(axis direction cs:\lengthWheel,\heightWheel); 
                \draw[draw = black, fill = black, rotate around={#6:(\FrontAxisX,\wheelRyC)}] (axis cs:\wheelxF,\wheelRy) rectangle ++(axis direction cs:\lengthWheel,\heightWheel); 
                \draw[line width=\linewidthOrient, -{Triangle[angle=60:2mm]}, colorPlot] (0,0) -- ++(axis direction cs:\lengthO,0); 
                \node [black] at (axis cs:0,0) {\scalebox{.5}{\textbullet}}; 
                \end{scope}
        \end{scope}
	
		}
	}
}

\tikzset{
    pics/Trailer/.style n args={9}{
		code = { %
        \pgfmathsetmacro\shiftChassisX{-0.2*#6} 
        \pgfmathsetmacro\shiftChassisY{-0.5*#7} 
        \pgfmathsetmacro\lengthWheel{0.2*#6} 
        \pgfmathsetmacro\heightWheel{0.075*#7} 
        \pgfmathsetmacro\wheelxR{-0.5*\lengthWheel} 
        \pgfmathsetmacro\wheelLy{0.2*#7+0.5*\heightWheel} 
        \pgfmathsetmacro\wheelRy{-0.2*#7-0.15*#7} 
        \pgfmathsetmacro\wheelFrontX{0.5*#6} 

        \pgfmathsetmacro\distTr{#8} 
        \pgfmathsetmacro\lengthTr{0.5*#6} 
        \pgfmathsetmacro\wheelFrontXT{0.35*\lengthTr} 
        \pgfmathsetmacro\OrientationTA{-#3+#4} 
        \pgfmathsetmacro\OrientationTB{-#4+#5}

        \begin{scope}[shift={(axis direction cs:#1,#2)}, rotate around={#3:(axis cs:0,0)}] 
            \begin{scope}[] 
                \begin{scope}[transparency group,opacity= #9] 
                \draw[draw = black, fill=grayCar] (axis cs:\shiftChassisX,\shiftChassisY) rectangle ++(axis direction cs:#6,#7); 
                \draw[very thin] (0,\wheelRy) -- (0,\wheelLy); 
                \draw[draw = black, fill = black] (axis cs:\wheelxR,\wheelLy) rectangle ++(axis direction cs:\lengthWheel,\heightWheel); 
                \draw[draw = black, fill = black] (axis cs:\wheelxR,\wheelRy) rectangle ++(axis direction cs:\lengthWheel,\heightWheel); 
                \node [black] at (axis cs:\wheelFrontX,0) {\scalebox{.2}{\textbullet}}; 
                \node [black] at (axis cs:0,0) {\scalebox{.3}{\textbullet}}; 
                \begin{scope}[rotate around={\OrientationTA:(axis cs:0,0)}, shift = {(axis direction cs:-\distTr,0)}]
                    \draw[draw = black, fill=grayCar] (axis cs:\shiftChassisX,\shiftChassisY) rectangle ++(axis direction cs:\lengthTr,#7); 
                    \draw[very thin] (0,\wheelRy) -- (0,\wheelLy); 
                    \draw[draw = black, fill = black] (axis cs:\wheelxR,\wheelLy) rectangle ++(axis direction cs:\lengthWheel,\heightWheel); 
                    \draw[draw = black, fill = black] (axis cs:\wheelxR,\wheelRy) rectangle ++(axis direction cs:\lengthWheel,\heightWheel); 
                    \node [black] at (axis cs:0,0) {\scalebox{.3}{\textbullet}}; 
                    \draw[very thin] (axis cs:0,0) -- (\distTr,0); 
                    \begin{scope}[rotate around={\OrientationTB:(axis cs:0,0)}, shift = {(axis direction cs:-\distTr,0)}]
                        \draw[draw = black, fill=grayCar] (axis cs:\shiftChassisX,\shiftChassisY) rectangle ++(axis direction cs:\lengthTr,#7); 
                        \draw[very thin] (0,\wheelRy) -- (0,\wheelLy); 
                        \draw[draw = black, fill = black] (axis cs:\wheelxR,\wheelLy) rectangle ++(axis direction cs:\lengthWheel,\heightWheel); 
                        \draw[draw = black, fill = black] (axis cs:\wheelxR,\wheelRy) rectangle ++(axis direction cs:\lengthWheel,\heightWheel); 
                        \node [black] at (axis cs:0,0) {\scalebox{.3}{\textbullet}}; 
                        \draw[very thin] (axis cs:0,0) -- (\distTr,0); 
                    \end{scope}
                \end{scope}
                \end{scope}
            \end{scope}
        \end{scope}

		}
	}
}

\title{Model Predictive Control of Non-Holonomic Vehicles:\\ Beyond Differential-Drive} 
\author{Mario Rosenfelder, Henrik Ebel\thanks{The corresponding author is H.~Ebel. \newline \textit{Email addresses:}\newline
\texttt{mario.rosenfelder@itm.uni-stuttgart.de} (Mario
Rosenfelder), \texttt{henrik.ebel@itm.uni-stuttgart.de}
(Henrik Ebel), \texttt{st160910@stud.uni-stuttgart.de} (Jasmin
Krauspenhaar), \texttt{peter.eberhard@itm.uni-stuttgart.de}
(Peter Eberhard)} , Jasmin Krauspenhaar, and Peter Eberhard \vspace{2ex}\\
    Institute of Engineering and Computational Mechanics, \\
    University of Stuttgart, Pfaffenwaldring 9, 70569 Stuttgart, Germany
}
\date{}
\maketitle
\begin{abstract}
Non-holonomic vehicles are of immense practical value and increasingly subject to automation. 
However, controlling them accurately, e.g., when parking, is known to be challenging for automatic control methods, including model predictive control~(MPC). 
Combining results from MPC theory and sub-Riemannian geometry in the form of homogeneous nilpotent system approximations, this paper proposes a comprehensive, ready-to-apply design procedure for MPC controllers to steer controllable, driftless non-holonomic vehicles into given setpoints. 
It can be ascertained that the resulting controllers nominally asymptotically stabilize the setpoint for a large-enough prediction horizon. 
The design procedure is exemplarily applied to four vehicles, including the kinematic car and a differentially driven mobile robot with up to two trailers.  
The controllers use a non-quadratic cost function tailored to the non-holonomic kinematics. 
Novelly, for the considered example vehicles, it is proven that a quadratic cost employed in an otherwise similar controller is insufficient to reliably asymptotically stabilize the closed loop. 
Since quadratic costs are the conventional choice in control, this highlights the relevance of the findings. 
To the knowledge of the authors, it is the first time that MPC controllers of the proposed structure are applied to non-holonomic vehicles beyond very simple ones, in particular (partly) on hardware. 
\end{abstract}

\section{Introduction}\label{sec:Intro}
Non-holonomic vehicles are omnipresent in technical applications. 
Cars, trucks, differentially driven service robots, and vehicles with trailers are nowadays indispensable for personal mobility, logistics, and automation. 
There are even applications of vehicles with multiple trailers, such as baggage tugs at airports. 
However, automatic control of such systems can be quite difficult, even compared to systems that look much more intricate to the uninitiated spectator. 
A key reason for that is that arguments made on the basis of a linearization of such systems do not lead very far since, even locally, controllability is lost in the process of linearization. 
While also presenting a challenge to humans wanting to control or park such vehicles, it challenges control methods such as model predictive control~(MPC). 
Often enough, in the latter's case, proofs of nominal stability do rely on some, albeit local, arguments or auxiliary controllers retrieved from controllable system linearizations, for instance in the case of the most well-known variants of model predictive control relying on a terminal set, e.g.,~\cite{ChenAllgoewer98}. 
As another symptom of the difficulty,  as proven in~\cite{MuellerWorthmann17}, even the arguably simplest non-holonomic vehicle, the differentially driven mobile robot, cannot be reliably stabilized asymptotically by means of an unconstrained nonlinear model predictive controller with quadratic cost, irrespective of the prediction horizon's length. 
This is crucial since, in applications, MPC controllers almost universally rely on a quadratic cost, penalizing in particular the squared and weighted Euclidean distance of the state or output from the corresponding desired value. 
But if this conventional choice does not work reliably, it is suddenly very unclear how to design a working MPC controller since the application engineer cannot rely on the usual paradigm of using a quadratic cost and enlarging the prediction horizon until a working controller is obtained. 
Still, despite the theoretical difficulties raised by deriving an MPC controller for these kinds of systems, the effort is worth it. 
This is since MPC controllers are rather easy to apply and tune. 
If the theoretical struggles of nonlinear MPC with non-holonomic vehicles are solved, one is left with a rather elegant solution for precise maneuvering of such vehicles, e.g., to park them accurately by merely repeatedly solving an optimal control problem~(OCP) with a suitable optimization algorithm. 
In comparison, other control techniques for the setpoint control of non-holonomic vehicles are either very heuristic and hence hard to maintain and tune or much more specific to a single type of vehicle and very arduous to derive, e.g., control using sinusoids~\cite{TeelMurrayWalsh95}. 
Also, other techniques rarely allow for a direct, non-conservative consideration of input constraints, which MPC does. 
Moreover, while some of the theoretical background can be arduous, the resulting OCPs are rather easy to understand and implement. 
Their structure is comparatively simple since, e.g., no additional constraints are introduced for nominal stabilization purposes, making easier implementation and numerical solution. 
In addition, since this paper demonstrates the proposed control design procedure in multiple examples, application-oriented readers may apply the procedure to other systems without studying every detail of the underlying theory, which is mostly necessary to connect the findings to theory that can guarantee a functioning controller and to obtain insightful mechanical interpretations. 

This paper is by no means the first paper aiming to overcome the difficulties of the nonlinear model predictive control of non-holonomic vehicles. 
Importantly, in~\cite{WorthmannEtAl16,WorthmannEtAl15}, the authors propose a non-quadratic cost function that achieves the stabilization task for the differentially driven mobile robot, proving that asymptotic stability is achieved for a certain minimum prediction horizon, without any additional design ingredients such as a terminal constraint or cost. 
Therein, proof and derivation are very technical, without much system theoretic or mechanical interpretation given, leaving little hope to directly generalize the approach to more intricate non-holonomic vehicles. 
However, in~\cite{RosenfelderEbelEberhard21}, we notice that the setup of the cost function proposed in~\cite{WorthmannEtAl16,WorthmannEtAl15} and our generalization to arbitrary, non-zero setpoints suggest a mechanical interpretation based on Lie brackets and directions of non-holonomic constraints. 
The key observation made in~\cite{RosenfelderEbelEberhard21} based on~\cite{WorthmannEtAl16,WorthmannEtAl15} is that different exponents are employed for the cost terms, so that, close to the setpoint, the cost function penalizes errors more in those directions that are harder to control, i.e., directions not covered directly by the input vector fields but only by vectors obtained from their Lie brackets. 
Intuitively, close to the setpoint, this allows the MPC controller to perform more extensive motions into easily controllable directions to compensate smaller deviations in directions that are not directly actuated. 
A human driver does the same when parking a non-holonomic vehicle. 
This intuitive human behavior when parking in parallel does not correspond to a cost-minimal behavior in the Euclidean sense; after all, that one would correspond to a motion made impossible by the non-holonomic constraint. 
Extensive numerical convergence studies in our previous work show that this interpretation enables us to come up with what seems like an asymptotically stabilizing non-quadratic cost function also for a system more complicated than the differentially driven robot, namely the kinematic car~\cite{RosenfelderEbelEberhard21}. 
However, our previous work presents mostly conjectures and empirical evidence and, thus, leaves many open questions regarding the formalization of the relationship between mechanical interpretation and control theory. 
The purpose of this paper is to close this gap. 

In this context, the paper makes the following contributions. 
First, to underline the necessity and relevance of the proposed MPC design procedure, we establish a sufficient condition for when quadratic costs are insufficient to asymptotically stabilize underactuated, input-affine driftless systems with a continuous-time nonlinear model predictive controller without terminal constraints or cost. 
This condition is fulfilled for all application examples studied in this paper, meaning that non-quadratic costs are actually necessary to reliably asymptotically stabilize those systems locally. 
In its generality, to the knowledge of the authors, the proposed condition is novel. 
Subsequently, we propose a novel design procedure for nonlinear model predictive controllers for general non-holonomic vehicles modeled kinematically as a driftless system. 
The procedure provides non-quadratic cost functions of a very specific setup.  
Controllers obtained that way can be proven to locally asymptotically stabilize the system if the controller's prediction horizon is long enough. 
Thirdly, we apply and delineate the proposed design procedure for four exemplary non-holonomic vehicles, including differentially driven robots with one and even with two trailers.  
In particular, these examples shall empower readers to repeat the procedure for other systems since some of the definitions employed seem rather arduous on paper but are easier to grasp when illustrated through applications. 
To the knowledge of the authors, it is the very first time that MPC controllers with the proposed kind of cost functions are applied to the more intricate kinds of non-holonomic vehicles considered. 
Moreover, we do not stop at theoretic deliberations or simulation results but also verify some results in the form of hardware experiments, in particular involving a differentially driven mobile robot with a trailer.  

Key theoretical foundations employed in this paper are the stability proofs in~\cite{CoronGrueneWorthmann20} for driftless non-holonomic systems for which a homogeneous approximation is known. 
We combine this with theory from~\cite{Jean14}, which uses results from sub-Riemannian geometry~\cite{Bellaiche96} to produce a procedure to obtain a specific kind of homogeneous system  approximation. 
Readers unfamiliar with homogeneous system approximations can see them as a tool as fundamental to non-holonomic systems as linearization to holonomic ones. 
The motivation for using homogeneous approximations and not a simpler, linear approximation is that controllability is lost in the process of linearization, making linearization useless for this paper's purposes. 

The paper is organized as follows. 
Section~\ref{sec:quadratic_cost} derives a condition under which quadratic costs are insufficient for the given setup. 
Section~\ref{sec:homog_approx} delineates how homogeneous approximations of the considered system class can be obtained in a constructive manner since these are useful to formulate OCPs leading to functioning MPC controllers for non-holonomic vehicles. 
It is then the topic of Section~\ref{sec:MPC_using_hom_approx} to show how these OCPs and the corresponding MPC controllers are generally furnished. 
These findings are applied in Section~\ref{sec:MPC_for_nonhol_vehicles} to obtain MPC controllers for various non-holonomic vehicles, including the differentially driven mobile robot (for arbitrary setpoints), the kinematic car, and the differentially driven robot with one and two trailers.  
To show the performance of the obtained controllers, in cases where hardware is available to us, we mainly look at hardware results, whereas simulations are considered where this is not the case. 

\section{Quadratic costs and driftless, input-affine non-holonomic vehicles}\label{sec:quadratic_cost}
In this paper, we look at controllable systems of the form $\dot{\bm{x}}=\bm{G}(\bm{x})\bm{u}$ with the control input~$\bm{u}\in\mathbb{R}^{n_u}$ and the system state~$\bm{x}\in\mathbb{R}^{n_x}$ with~$n_u<n_x$ so that the system is underactuated. 
Kinematic models of non-holonomic vehicles such as the differentially driven mobile robot, potentially with trailer(s), and the kinematic car are all of this form. 
The goal of this section is to show when quadratic costs do not work properly for the considered systems, serving as motivation for the remaining paper. 
To that end, in the following, we follow the argumentation from~\cite{MuellerWorthmann17} but instead of only looking at the differentially driven mobile robot, we extend the arguments to the more general above-mentioned class of systems. 
Without loss of generality, in this section, the origin~$\bm{0}$ shall be stabilized asymptotically. 
However, before deriving a condition for the insufficiency of quadratic costs, the introduction of some notation and a clear definition of the considered MPC problem are necessary. 

We employ the usual MPC receding horizon strategy, meaning that, at each time instant~$t\coloneqq k\,\delta_t$, for a given sampling time~$\delta_t> 0$, the OCP
\begin{subequations}
    \begin{alignat}{3}%
            &\hspace{0pt} \underset{\bm{u}(\cdot\,\vert\, t)}{\text{minimize}}
            &&\hspace{6pt}\!\int_{t}^{t+T}\!\ell(\bm{x}(\tau\,\vert\, t),\bm{u}(\tau\,\vert\, t))\; \textnormal{d}\tau \label{eq:mpc_cost_generic}\\
            & \text{subject to}
            &&\hspace{6pt}\dot{\bm{x}}(\tau\,\vert\, t)=\bm{G}(\bm{x}(\tau\,\vert\, t))\bm{u}(\tau\,\vert\, t),\label{eq:mpc_dynamics_generic}\\
            &&&\hspace{6pt} \bm{u}(\tau \,\vert\, t)\in\mathcal{U}\quad \forall \tau \in [t,t+T),\\
            &&&\hspace{6pt}\bm{x}(t\,\vert\, t)=\bm{x}(t)
    \end{alignat}%
    \label{eq:mpc_optprob_generic}%
\end{subequations}%
is solved with the stage cost~$\ell(\bm{x},\bm{u})$. 
Therein, we use the usual MPC notation in which~$\bm{u}(\tau\,\vert\, t)$ is the control input trajectory planned at time~$t$ and evaluated at times~$\tau\in [t,t+T)$ for a given closed input constraint set~$\mathcal{U}\subset\mathbb{R}^{n_u}$. 
Using the planned input trajectory, the predicted state trajectory~$\bm{x}(\cdot\,\vert\, t)$ results from integrating the system dynamics~\eqref{eq:mpc_dynamics_generic} over the prediction horizon. 
Having solved the optimization problem, yielding the optimal solution~$\bm{u}^\star(\cdot \,\vert\, t)$, the optimal control input is applied on the time interval~$[t,t+\delta_t)$, i.e., $\bm{u}({\tau})\coloneqq \bm{u}^\star(\tau \,\vert\, t)$ for $\tau\in[t,t+\delta_t)$, and the OCP is solved anew at time~$t+\delta_t$. 

In this section, it is assumed that the cost function~\eqref{eq:mpc_cost_generic} is quadratic, and, thus, it is assumed that the stage cost takes the form~$\ell(\bm{x},\bm{u})\coloneqq\bm{x}\tran\bm{Q}\bm{x}+\bm{u}\tran\bm{R}\bm{u}$ with symmetric, positive definite weighting matrices~$\bm{Q}, \bm{R}$, i.e., $\bm{Q},\bm{R}\succ \bm{0}$. 
This choice of cost is insufficient if, irrespective of the length of the optimization horizon~$T\in\mathbb{R}_{> 0}$, there always exists an initial condition~$\bm{x}(0)\eqqcolon \bm{x}_0 \neq \bm{0}$ such that~$\bm{u}^{\star}(\cdot\,\vert\, 0)\equiv \bm{0}$ is the unique optimal solution of the MPC optimization problem. 
In fact, the subsequently derived condition requires that such an initial condition exists in each set~$\mathcal{B}_{\varepsilon}(\bm{0})\setminus \{\bm{0}\}$, with~$\mathcal{B}_{\varepsilon}(\bm{p})\coloneqq\left\lbrace \bm{x}\in\mathbb{R}^{n_x}\ \vert\ \left\Vert \bm{x}-\bm{p} \right\Vert_2 \leq \varepsilon \right\rbrace$, irrespective of the value of~$\varepsilon>0$. 
For the given class of dynamical systems, this means that the application of the first part of the optimal input sequence lets the system stay stationary in~$\bm{x}_0\neq \bm{0}$, as do subsequent solutions of the MPC problem following the receding horizon strategy. 
Therefore, the system state will remain in~$\bm{x}_0\neq\bm{0}$ for eternity and, in particular, the controller does not asymptotically stabilize~$\bm{0}$. 
In the first step, the following lemma introduces a condition under which the zero input fulfills the necessary optimality conditions. 
\begin{lem}
    \label{lem:nec_opt_cond}
    Assume that the prediction horizon~$T\in (0,\infty)$ and weighting matrices~$\bm{Q}=\bm{Q}\tran \succ \bm{0}$,~$\bm{R}=\bm{R}\tran \succ \bm{0}$ are given. 
    Then, if the initial condition~$\bm{x}_0\in\mathbb{R}^{n_x}\setminus \lbrace \bm{0} \rbrace$ satisfies
    \begin{equation}
        \bm{x}_{0}\tran\bm{Q}\bm{G}(\bm{x}_0)=\bm{0},
        \label{eq:lemma_1}
    \end{equation}
    the control input trajectory~$\bm{u}^{\star}(\cdot\,\vert\, 0)\equiv \bm{0}$, which is constantly zero over the prediction horizon, satisfies the necessary optimality conditions of the MPC optimization problem~\eqref{eq:mpc_optprob_generic}.
\end{lem}
\begin{proof}
    The goal is to show the three necessary optimality conditions from Pontryagin's maximum principle~\cite[Section~2.2.3]{Sethi19}, similar to~\cite{MuellerWorthmann17}. 
    For notational brevity, in this proof, we write~$\bm{u}^{\star}(\cdot)$ instead of~$\bm{u}^{\star}(\cdot\,\vert\, 0)$ and~$\bm{x}^{\star}(\cdot)$ instead of the corresponding predicted state trajectory~$\bm{x}^{\star}(\cdot\,\vert\, 0)$. 
    Furthermore, time dependencies are omitted notation-wise if not leading to ambiguities. 
    To this end, the (control) Hamiltonian is defined as~$\mathcal{H}(\bm{x},\bm{u},\bm{\lambda})\coloneqq -\bm{x}\tran\bm{Q}\bm{x}-\bm{u}\tran\bm{R}\bm{u}+\bm{\lambda}\tran\bm{G}(\bm{x})\bm{u}$ with the adjoint variable~$\bm{\lambda}\in\mathbb{R}^{n_x}$. 
    The first optimality condition reads $\dot{\bm{x}}^{\star}(t)\overset{\scriptscriptstyle !}{=}\frac{\partial}{\partial\bm{\lambda}}\mathcal{H}\tran (\bm{x}^{\star},\bm{u}^{\star},\bm{\lambda}^{\star})$. 
    Since~$\frac{\partial}{\partial\bm{\lambda}}\mathcal{H}\tran \vert_{(\bm{x}^{\star},\bm{u}^{\star},\bm{\lambda}^{\star})}=\bm{G}(\bm{x}^{\star})\bm{u}^{\star}$, the condition is trivially satisfied due to the predicted state and input trajectories jointly satisfying the system dynamics by design of the MPC problem. 
    The second condition reads~$\dot{\bm{\lambda}}^{\star}\overset{\scriptscriptstyle !}{=}-\frac{\partial}{\partial \bm{x}}\mathcal{H}\tran\vert_{(\bm{x}^{\star},\bm{u}^{\star},\bm{\lambda}^{\star})}$ with the right-hand side yielding~$-\frac{\partial}{\partial \bm{x}}\mathcal{H}\tran\vert_{(\bm{x}^{\star},\bm{u}^{\star},\bm{\lambda}^{\star})}=2\,\bm{Q}\bm{x}^{\star}-{\bm{\lambda}^{\star}}\tran\frac{\partial\bm{G}(\bm{x})}{\partial \bm{x}}\big\vert_{\bm{x}^{\star}}\bm{u}^{\star}=2\,\bm{Q}\bm{x}_0$ for~$\bm{u}^{\star}\equiv \bm{0}$. 
    Furthermore, the transversality condition reads, in the absence of a terminal cost,~${\bm{\lambda}}^{\star}(T)\overset{\scriptscriptstyle !}{=}\bm{0}$. 
    Thus, the second condition is satisfied together with the transversality condition for~$\bm{\lambda}^{\star}(t)\coloneqq 2\,(t-T)\bm{Q}\bm{x}_0$. 
    The third condition reads $\mathcal{H}(\bm{x}^{\star}, \bm{u}^{\star}, \bm{\lambda}^{\star})\overset{\scriptscriptstyle !}{\geq}\mathcal{H}(\bm{x}^{\star}, \bm{u}, \bm{\lambda}^{\star})\ \forall \bm{u}\in\mathcal{U}$. 
    The left-hand side can be directly evaluated to~$\mathcal{H}(\bm{x}^{\star}, \bm{u}^{\star}, \bm{\lambda}^{\star})=-\bm{x}_0\tran \bm{Q}\bm{x}_0$, whereas the right-hand side yields~$\mathcal{H}(\bm{x}^{\star}, \bm{u}, \bm{\lambda}^{\star})=-\bm{x}_0\tran\bm{Q}\bm{x}_0-\bm{u}\tran\bm{R}\bm{u}+2(t-T)\bm{x}_0\tran\bm{Q}\bm{G}(\bm{x}_0)\bm{u}$. 
    Inserting both into the third condition and simplifying finally leads to
    \begin{align}
        \bm{u}\tran\bm{R}\bm{u}\overset{\scriptscriptstyle !}{\geq}2(t-T)\bm{x}_0\tran\bm{Q}\bm{G}(\bm{x}_0)\bm{u}\overset{\eqref{eq:lemma_1}}{=}\bm{0},
    \end{align}
    which is true due to positive definiteness of~$\bm{R}$. 
\end{proof}

So far, it has been shown that, under condition~\eqref{eq:lemma_1}, the zero control input satisfies necessary first-order optimality conditions. However, in order to conclude the failure of the controller subject to quadratic costs it remains to be proven that this solution is indeed the global optimum. This is shown in the following theorem using sufficient second-order conditions as well as uniqueness of the obtained optimal solution.
\begin{thm}
    Assume that the prerequisites of Lemma~\ref{lem:nec_opt_cond} are fulfilled. 
    Moreover, assume that there exists an $\bm{x}_0\in\mathcal{B}_{\varepsilon}(\bm{0})\setminus \{\bm{0}\}$ for all~$\varepsilon>0$ that satisfies Eq.~\eqref{eq:lemma_1}. 
	Then, for any such initial condition~$\bm{x}_0$, $\bm{u}^\star(\cdot\,\vert\,0)\equiv \bm{0}$ is the unique optimal control input.
    \label{thm:quadratic_cost_insufficient}
\end{thm}
\begin{proof}
	Once again, this proof can be seen as a generalized version of the one given in~\cite{MuellerWorthmann17} for differentially driven mobile robots based on~\cite{Sethi19}. 
    The idea is to show optimality of~$\bm{u}^\star(\cdot\,\vert\,0)\equiv \bm{0}$ and uniqueness of the optimizer. 
    The proof for the latter is actually given in a general manner in~\cite{MuellerWorthmann17} and also applies for this paper's class of systems. 
    Hence, optimality remains to be proven. 
    According to~\cite[Theorem~2.1]{Sethi19}, for~$\bm{u}^{\star}$, $\bm{x}^{\star}$, $\bm{\lambda}^{\star}$ jointly satisfying the necessary optimality conditions already shown to hold in Lemma~\ref{eq:lemma_1} for~$\bm{u}^{\star}\equiv \bm{0}$,~$\bm{u}^{\star}$ is optimal if the function~$\mathcal{H}_0(\bm{x},\bm{\lambda}^{\star})\coloneqq \max_{\bm{u}}\mathcal{H}(\bm{x},\bm{u},\bm{\lambda}^{\star})$ is concave in~$\bm{x}$ at all times~$t\in [0,T]$ (checking this one condition suffices if there is no terminal cost). 
    To obtain the Hamiltonian's optimizer to compute~$\mathcal{H}_0$, we set~$\frac{\partial }{\partial \bm{u}}\mathcal{H}\overset{\scriptscriptstyle !}{=}\bm{0}$, yielding $\frac{\partial }{\partial \bm{u}}\mathcal{H}=-2\,\bm{u}\tran\bm{R}+{\bm{\lambda}^{\star}}\tran\bm{G}\overset{\scriptscriptstyle !}{=}\bm{0}$, and, using~$\bm{R}\succ \bm{0}$,~$\bm{u}\tran = \frac{1}{2} {\bm{\lambda}^{\star}}\tran \bm{G}\bm{R}^{-1}$. 
    The latter is the unique optimizer since~$\mathcal{H}$ is strictly concave with respect to~$\bm{u}$. 
    Hence, inserting the optimizer yields~$\mathcal{H}_0(\bm{x},{\bm{\lambda}^{\star}})=-\bm{x}\tran\bm{Q}\bm{x}+\frac{1}{4}{\bm{\lambda}^{\star}}\tran\bm{G}\bm{R}^{-1}\bm{G}\tran{\bm{\lambda}^{\star}}$ and concavity can be established by inspecting
    \begin{align}
        \frac{\partial^2}{\partial \bm{x}^2}\mathcal{H}_0 = -2\,\bm{Q}&+\frac{1}{2}{\bm{\lambda}^{\star}}\tran \frac{\partial^2 \bm{G}}{\bm{x}^2}\bm{R}^{-1}\bm{G}\tran{\bm{\lambda}^{\star}}\nonumber\\{}&+\frac{1}{2}{\bm{\lambda}^{\star}}\tran\frac{\partial \bm{G}}{\partial \bm{x}}\bm{R}^{-1}\frac{\partial \bm{G}}{\partial \bm{x}}{\bm{\lambda}^{\star}}.
        \label{eq:hamiltonian_hessian}
    \end{align}
    Since, according to the proof of Lemma~\ref{lem:nec_opt_cond}, the adjoint variable is given by $\bm{\lambda}^{\star}(t)= 2\,(t-T)\bm{Q}\bm{x}_0$, for all~$t\in[0, T]$, the entries of the adjoint variable become arbitrarily small if~$\bm{x}_0$ is chosen arbitrarily close to the origin. 
    The latter is possible due to the assumptions. 
    Since the right-hand side of~\eqref{eq:hamiltonian_hessian} is continuous in the adjoint variable, and, thus, with the eigenvalues being a continuous function of the adjoint variable~\cite[Corollary VI.1.6]{Bhatia97}, for~$\bm{x}_0$ chosen sufficiently close to the origin, the right-hand side of~\eqref{eq:hamiltonian_hessian} is negative definite since~$\bm{Q}$ is positive definite.
\end{proof}
Hence, the assumptions of Theorem~\ref{thm:quadratic_cost_insufficient} form a sufficient condition on when, for the considered class of systems, an MPC controller based on OCP~\eqref{eq:mpc_optprob_generic} with conventional quadratic stage costs does not function properly, meaning that, in every neighborhood of the origin, there exists an initial state for which the controlled system will not even begin to drive toward the origin. 
While the crucial condition~\eqref{eq:lemma_1} seems abstract at first, it has a clear mechanical interpretation. 
Namely, if the weighting matrix~$\bm{Q}$ is chosen as the identity, the condition holds for states which have the same direction in the state space as a non-holonomic constraint of the system evaluated at this point. 
To show the assertion for an arbitrary symmetric and positive definite weighting matrix~$\bm{Q}$ accordingly yields that the direction of the state and the non-holonomic constraint does not need to coincide, although the linear mapping conveyed by the matrix~$\bm{Q}$ merely rotates and scales the situation.  
For the sake of completeness, despite not being necessary to understand the remainder of the paper, the technical calculations in App.~\ref{app:proof} give a sketch on how to show that the assumptions of Theorem~\ref{thm:quadratic_cost_insufficient} are satisfied for popular non-holonomic vehicles. 
However, if that is the case, this leaves the question of how to design a functioning OCP, i.e., in particular its cost function. 
Answers are found in the subsequent sections. 
Since usual first-order approximations such as linearization lead to a critical loss of controllability at the desired setpoint of the non-holonomic systems in question, approaches build on a specific kind of homogeneous approximation, which therefore is looked at first in the following section. 
Later, this approximation is used in a straightforward fashion to construct the tailored non-quadratic cost function in Sec.~\ref{sec:MPC_using_hom_approx}. 

\section{Homogeneous approximation of non-holonomic systems}
\label{sec:homog_approx}
In order to analyze and suitably approximate the considered non-holonomic systems, the individual columns of the input matrix are crucial, motivating to write 
\begin{align}\label{eq:system}
\dot{\bm{x}} = \bm{G}(\bm{x})\bm{u} \eqqcolon \begin{bmatrix} \bm{X}_1(\bm{x}) & \cdots & \bm{X}_{n_u}(\bm{x})\end{bmatrix} \bm{u}, 
\end{align} where $\bm{X}_i(\bm{x}),~i\in\mathbb{Z}_{1:n_u},$ are $C^\infty$ controlling vector fields on $\mathbb{R}^{n_x}$. 
Note that any differentiable manifold can be considered, however, we restrict our investigations without loss of generality to $\mathbb{R}^{n_x}$ in the following. 
The controlling vector fields generate a family of vector spaces~$\bm{\Delta}^1(\bm{x})=\text{span}\{\bm{X}_1(\bm{x}), \ldots, \bm{X}_{n_u}(\bm{x})\}\subset T_{\bm{x}}\mathbb{R}^{n_x}$. 
If the dimension of~$\bm{\Delta}^1(\bm{x})$ is constant over all~$\bm{x}\in\mathbb{R}^{n_x}$ (which is the case in this paper's applications),~$\bm{x}\mapsto \bm{\Delta}^1(\bm{x})$ is a distribution on~$\mathbb{R}^{n_x}$ in the differential geometric sense~\cite{Jean14}. 
In the following, for brevity, we will mostly omit the arguments of vector fields and of quantities calculated from them. 
Recursively, following~\cite{Jean14}, we define $\bm{\Delta}^{s+1} \coloneqq \bm{\Delta}^{s} \cup [\bm{\Delta}^1,\,\bm{\Delta}^s],~s\in\mathbb{Z}_{\geq 1}$, with the addend reading $[\bm{\Delta}^1,\,\bm{\Delta}^s]\coloneqq\text{span}\left\lbrace [\bm{v}_1,\,\bm{v}_2]:\, \bm{v}_1 \in \bm{\Delta}^1, \, \bm{v}_2\in\bm{\Delta}^s \right\rbrace$, i.e., following from Lie bracket operations. 
The union set of all~$\bm{\Delta}^{s}$ is the so-called Lie algebra spanned by the control vector fields, i.e., $\text{Lie}\left(\bm{X}_1,\ldots,\bm{X}_{n_u} \right)=\cup_{s\in\mathbb{Z}_{\geq1}} \bm{\Delta}^s$~\cite{Jean14}. 
The Lie algebra can then be utilized to determine whether a system is controllable or not, e.g., via the Ra\v{s}evskij-Chow theorem or the Lie algebra rank condition~(LARC). 
We assume that the non-holonomic system at hand is controllable which is indeed the case for all non-holonomic vehicles within the scope of this article. It is important to note that the Lie algebra rank condition suffices to show global controllability of driftless, control-affine systems~\eqref{eq:system}, see~\cite[Ch.~3.3]{Coron07}. 
Subsequently, the degree of non-holonomy~$r\in\mathbb{Z}_{\geq 1}$ at~$\bm{x}\in\mathbb{R}^{n_x}$ of a controllable, non-holonomic system~\eqref{eq:system} is defined as the smallest integer such that~$\bm{\Delta}^{r}(\bm{x})=T_{\bm{x}}\mathbb{R}^{n_x}$, i.e., $\textnormal{dim}\left(\bm{\Delta}^{r}(\bm{x})\right)=n_x$.  
Thus, $r$ is a measure of how difficult it is to steer the underactuated system~\eqref{eq:system} to arbitrary setpoints by interleaving its input vector fields. 

The approximation of a non-holonomic system~\eqref{eq:system} in a way such that its crucial structural properties, particularly its controllability, are retained is not a trivial task. It can be shown that linearizing the system results in the loss of controllability~\cite{Jean14}, such that the linearized system does not represent the local behavior of the original nonlinear system. 
Following~\cite{Jean14}, the linearization is a first-order approximation of nonlinear systems in Euclidean geometry while non-holonomic systems usually must be handled with sub-Riemannian geometry. 
Therefore, the first-order approximation of these systems should be done with respect to sub-Riemannian geometry, motivating so-called homogeneous approximations. 
Homogeneous systems, which can be seen as a generalization of linear systems, have received significant attention in the field of nonlinear control theory~\cite{Coron07,Gruene00,Jean14}. Following~\cite{CoronGrueneWorthmann20}, for given tuples~$\bm{r}=(r_1,\ldots,r_{n_x})\in(0,\infty)^{n_x}$, $\bm{s}=(s_1,\ldots,s_{n_u})\in(0,\infty)^{n_u}$, and $\tau\in(-\min_i r_i, \infty)$, system~\eqref{eq:system} is $(\bm{r},\bm{s},\tau)$-homogeneous if it holds that
\begin{align}
	\bm{G}(\bm{\Lambda}_\alpha \bm{x}) \bm{\Delta}_{\alpha} \bm{u} = \alpha^\tau \bm{\Lambda}_\alpha \bm{G}(\bm{x}) \bm{u}
\end{align} for all $\bm{x}\in\mathbb{R}^{n_x}$, $\bm{u}\in\mathbb{R}^{n_u}$, and $\alpha\in\mathbb{R}_{\geq0}$ with the diagonal matrices $\bm{\Lambda}_\alpha\coloneqq \text{diag}(\alpha^{r_1},\ldots,\alpha^{r_{n_x}})\in\mathbb{R}^{n_x \times n_x}$ and $\bm{\Delta}_\alpha\coloneqq \text{diag}(\alpha^{s_1},\ldots,\alpha^{s_{n_u}})\in\mathbb{R}^{n_u \times n_u}$. 
Notably, if the dynamics~\eqref{eq:system} is globally Lipschitz continuous, the degree of homogeneity is~$\tau=0$~\cite{Gruene00}.

However, unlike linearizing the system, it is far from trivial to derive a homogeneous approximation of general non-holonomic systems of the form~\eqref{eq:system}. 
One constructive approach is using so-called homogeneous nilpotent approximations. 
To this end, it is inevitable to first introduce some definitions and tools from sub-Riemannian geometry, mostly following~\cite{Jean14}.

\subsection{Preliminaries from sub-Riemannian geometry}
First-order approximations are in general concerned with respect to a fixed setpoint which is denoted as $\bm{d}\in\mathbb{R}^{n_x}$ subsequently. 
A Lie derivative of a function ${f}\in C^\infty$ along a vector field~$\bm{X}_i$, i.e., $\bm{X}_i {f}$, is also known as a first-order non-holonomic derivative of ${f}$ along $\bm{X}_i$. 
Consequently, $\bm{X}_i(\bm{X}_j {f})$ and $\bm{X}_i(\bm{X}_j(\bm{X}_k {f}))$ are non-holonomic derivatives of ${f}$ of order $2$ and $3$, respectively, along the respective vector fields, and so on. 
Fittingly,~${f}(\bm{d})$ is the non-holonomic derivative of order~$0$ at $\bm{d}$. 
The non-holonomic order of ${f}$ at $\bm{d}$, denoted by $\text{ord}_{\bm{d}}({f})$, is the largest integer~$k$ for which all non-holonomic derivatives of ${f}$ of order smaller than ${k}$ vanish at $\bm{d}$. 
Consequently, $\text{ord}_{\bm{d}}({f})=0$ holds if ${f}(\bm{d})\neq0$. 
Analogously, the non-holonomic order of a vector field $\bm{X}$ at $\bm{d}$, i.e., $\text{ord}_{\bm{d}}(\bm{X})$, is defined as the supremal $q\in\mathbb{R}$ such that it holds that $\text{ord}_{\bm{d}}(\bm{X}f)\geq q+\text{ord}_{\bm{d}} (f)$ for all $f\in C^\infty (\bm{d})$. For a more detailed insight regarding non-holonomic derivatives and orders, we refer the interested reader to~\cite{Jean14}.  
It can be shown that the control vector fields of~\eqref{eq:system} are of order $\geq -1$, the resulting Lie brackets $[\bm{X}_i, \, \bm{X}_j],~i,j\in\mathbb{Z}_{1:n_u}$, are of order $\geq -2$, and so on~\cite{Vendittelli04}. 
Based on~\cite{Vendittelli04}, we can also clarify the term nilpotent. 
The approximation of~\eqref{eq:system} is said to be nilpotent at $\bm{d}$, if the difference of the original and approximated vector fields are of order $\geq 0$ at $\bm{d}$ and its Lie algebra is nilpotent of step $s>r(\bm{d})$. 
The first property ensures the preservation of controllability~\cite{Vendittelli04}, whereas the latter means that all Lie brackets obtained through recursively nesting Lie brackets~$s$ times are zero~\cite{Jean14}. 
For instance, if the brackets of the form $[\bm{X}_i, \, [\bm{X}_j,\bm{X}_k]]$ evaluate to zero for all~$i,j,k\in\mathbb{Z}_{1:n_u}$, the corresponding Lie algebra is nilpotent of step~2 since two brackets are nested in the expression. 
Inductively, if the brackets of the form~$[\bm{X}_i,[\bm{X}_j, \, [\bm{X}_k,\bm{X}_\ell]]]$ evaluate to zero for all~$i,j,k,.\ell\in\mathbb{Z}_{1:n_u}$, the corresponding Lie algebra is nilpotent of step~3, and so forth. 

Further, with $n_i(\bm{d}) \coloneqq \text{dim}\left(\bm{\Delta}^i(\bm{d})\right)$, the $r$-tuple of integers $\bm{n}(\bm{d}) = \left( n_1 (\bm{d}),\ldots,n_{r}(\bm{d}) \right)$ is the so-called growth vector at $\bm{d}$ where it holds for the given system~\eqref{eq:system} that $n_1(\bm{d})=\textnormal{rank}(\bm{G}(\bm{d}))$ and $n_r(\bm{d})=n_x$ with~$r$ being the degree of non-holonomy~\cite{Jean14}. 
Subsequently, we assume that the non-holonomic system in question is only analyzed and approximated at so-called regular points~$\bm{d}$, meaning that the growth vector and, in particular, the degree of non-holonomy are constant locally around~$\bm{d}$~\cite{Vendittelli04}.
Following~\cite{Jean14}, we define the weights $w_i = w_i(\bm{d}),~i\in\mathbb{Z}_{1:n_x},$ at $\bm{d}$ as $w_j=s$ if $n_{s-1}(\bm{d})<j\leq n_s(\bm{d})$, where $n_0=0$. For brevity, we consolidate the weights in the $n_x$-tuple $\bm{w}(\bm{d})$. Using the weights, a so-called adapted frame at $\bm{d}$ is formed which is a family of vector fields $\bm{X}_1,\ldots,\bm{X}_{n_x}$ such that $\{\bm{X}_1(\bm{d}),\ldots,\bm{X}_{n_x}(\bm{d}) \}$ is a basis of the tangent space~$T_{\bm{d}}\mathbb{R}^{n_x}$ of~$\mathbb{R}^{n_x}$ at~$\bm{d}$, and $\bm{X}_i\in\bm{\Delta}^{w_i},~i\in\mathbb{Z}_{1:n_x}$~\cite{Jean14}. Moreover, the weights are then used to define a specific system of local coordinates which will be crucial for the construction of tailored MPC cost functions. One way to obtain these coordinates is dealt with subsequently.

\subsection{Privileged coordinates}
The so-called privileged coordinates~$\bm{z}$ at a point $\bm{d}$ are a system of local coordinates with maximum possible order, i.e., they satisfy $\text{ord}_{\bm{d}} (z_j)=w_j$ for all $j\in\mathbb{Z}_{1:n_x}$. These coordinates are vital to derive the desired homogeneous nilpotent approximation of~\eqref{eq:system}. One type of privileged coordinates are algebraic coordinates which can be be constructed, for example, using Bella\"{i}che's algorithm~\cite{Bellaiche96,Jean14}. 
This algorithm can be divided into two substeps. Starting point is an adapted frame $\{\bm{X}_1(\bm{d}),\ldots,\bm{X}_{n_x}(\bm{d}) \}$ following from the Lie algebra $\text{Lie}\left(\bm{X}_1,\ldots,\bm{X}_{n_u} \right)$.  
In the first step, the original coordinates $\bm{x}$ are transformed into the coordinates $\bm{y}$ via 
\begin{align}\label{eq:Trafo1Bellaiche}
	\bm{y} = \left.\begin{bmatrix} \bm{X}_1\tran \\ \vdots \\ \bm{X}_{n_x}\tran \end{bmatrix}\right|_{\bm{d}}^{-\mkern-1.5mu\mathsf{T}} (\bm{x}-\bm{d})\eqqcolon \bm{A}^{-\mkern-1.5mu\mathsf{T}} (\bm{x}-\bm{d}),
\end{align} 
For a derivation of this transformation we refer to~\cite{HrdinaEtAl2017}. 
In the second step of the algorithm, the system of privileged coordinates $\bm{z}$ follows from
\begin{align}\label{eq:Trafo2Bellaiche}
	&z_j = y_j - \sum_{k=2}^{w_j-1} h_{j,k} \left( \bm{y} \right), \quad j\in\mathbb{Z}_{1:n_x}, \nonumber\\
	&h_{j,k} (\bm{y}) =  \\
	&\sum_{\substack{|\bm{\alpha}|=k \\ w(\bm{\alpha}) < w_j}}\hspace{-10pt} \bm{Y}_1^{\alpha_1} \cdots \bm{Y}_{j-1}^{\alpha_{j-1}} \left.\left( y_j - \sum_{q=2}^{k-1} h_{j,q} (\bm{y})\right)\right|_{\bm{d}} \prod_{i=1}^{j-1} \frac{y_i^{\alpha_i}}{\alpha_i !}, \nonumber
\end{align} 
$k\in\mathbb{Z}_{2:w_j-1}$, with $\bm{\alpha}=(\alpha_1,\ldots,\alpha_{n_x})$ being a multi-index,~$\alpha_i\in\mathbb{N}_0$, $|\bm{\alpha}| = \sum_{i=1}^{n_x} \alpha_i$, and $w(\bm{\alpha})=\sum_{i=1}^{n_x} w_i \alpha_i$. 
Note that $\bm{Y}_i$ denotes the corresponding $i$th~vector field following from the transformation~\eqref{eq:Trafo1Bellaiche}, i.e., $\bm{Y}_k=\frac{\partial \bm{y}}{\partial\bm{x}}\bm{X}_k$. 
Moreover, $\bm{Y}_a^b (c)\vert_{\bm{d}}$ denotes the non-holonomic derivative of the function $c$ of order $b$ along the vector field $\bm{Y}_a$ evaluated at $\bm{d}$. 
Crucially, the function~$h_{j,k}$ actually only depends on the first $j-1$~entries of~$\bm{y}$, i.e., $h_{j,k}(\bm{y})=h_{j,k} (y_1,\ldots,y_{j-1})$. 
It is worth noting that the second step~\eqref{eq:Trafo2Bellaiche} yields~$\bm{y}=\bm{z}$ for the considered non-holonomic vehicles at the origin, greatly simplifying the derivation of the privileged coordinates. 
In particular, the previously presented holistic procedure based on sub-Riemannian geometry simplifies to finding an adapted frame and evaluating~\eqref{eq:Trafo1Bellaiche} for this case. 
However, this does not hold in general for other setpoints or systems.

\subsection{Homogeneous nilpotent approximation}
After having established the privileged coordinates~$\bm{z}$ of the system centered at the approximation point~$\bm{d}$, finally a homogeneous nilpotent approximation of~\eqref{eq:system} can be derived, again following~\cite{Jean14}. 
In particular, the control vector fields corresponding to the privileged coordinates can be expanded as a Taylor series $\bm{Z}_i(\bm{z})\propto \sum_{\bm{\alpha},j} a_{\bm{\alpha},j} \bm{z}^{\bm{\alpha}} \frac{\partial}{\partial z_j},~i\in\mathbb{Z}_{1:n_u}$, with $a_{\bm{\alpha},j}\in\mathbb{R}$. 
Therein, the so-called weighted degree of a monomial $\bm{z}^{\bm{\alpha}}=\prod_{k=1}^{n_x} z_k^{\alpha_k}$ is defined as $w(\bm{\alpha})\coloneqq\sum_{k=1}^{n_x} w_k \alpha_k$, with~$w_k$ being the weights as defined before, satisfying~$\textnormal{ord}_{\bm{d}}(z_j)=w_j$.  
Note that $\bm{\alpha}$ is another multi-index which does not correlate with the previously introduced multi-index corresponding to the second step~\eqref{eq:Trafo2Bellaiche} of Bella\"{i}che's algorithm. 
Then, the weighted degree of the monomial vector field $\bm{z}^{\bm{\alpha}}\frac{\partial}{\partial z_j}$ reads $w(\bm{\alpha})-w_j$~\cite{Jean14}. Grouping together the monomial vector fields having the same weighted degree, the control vector fields $\bm{Z}_i,~i\in\mathbb{Z}_{1:n_u}$, can be expressed as a series in terms of
\begin{align}\label{eq:vector_field_Taylor}
	\bm{Z}_i (\bm{z}) = \bm{Z}_i^{[-1]}(\bm{z}) + \bm{Z}_i^{[0]}(\bm{z}) + \bm{Z}_i^{[1]}(\bm{z}) + \dots,
\end{align} 
where $\bm{Z}_i^{[s]}$ is a vector field with weighted degree~$s$.
Truncating after the first term, the vector fields $\bm{Z}_1^{[-1]},\ldots,\bm{Z}_{n_u}^{[-1]}$ generate a Lie algebra that is nilpotent of step $w_{n_x}$ and, moreover, is a first-order approximation of the control vector fields $\bm{Z}_1,\ldots,\bm{Z}_{n_u}$ at $\bm{d}$~\cite{Jean14}. 
Thus, we obtain a homogeneous nilpotent approximation of the original system at $\bm{d}$ related to the privileged coordinates $\bm{z}$ reading
\begin{align}\label{eq:system_approx}
	\dot{\tilde{\bm{z}}} = \begin{bmatrix} \bm{Z}_1^{[-1]} (\tilde{\bm{z}}) & \cdots &  \bm{Z}_{n_u}^{[-1]} (\tilde{\bm{z}})  \end{bmatrix} \bm{u}.
\end{align} 
We emphasize that this is only an approximation of the original dynamics~\eqref{eq:system} at $\bm{d}$ by using the tilde symbol since~\eqref{eq:system} can also exactly be expressed in the privileged coordinates~$\bm{z}$. It is worth noting that the resulting non-holonomic approximation~\eqref{eq:system_approx} is  polynomial and in triangular form, i.e., $\dot{\tilde{z}}_j = \sum_{i=1}^{n_u} g_{ij}(\tilde{z}_1,\ldots,\tilde{z}_{j-1})u_i$~\cite{Vendittelli04}. 

As hinted at, due to being homogeneous, the system approximation~\eqref{eq:system_approx} is useful to construct a cost function for an MPC optimization problem geometrically tailored to the system's non-holonomic constraints. 
Exactly how this is furnished is the topic of the subsequent section. 
Afterwards, applications will give practical insight on how to apply the theory hitherto described. 

\section{Tailored MPC based on homogeneous approximation}
\label{sec:MPC_using_hom_approx}
The overall goal is to derive a tailored cost function for non-holonomic vehicles since the standard quadratic cost is insufficient in all following examples, cf.~Sec.~\ref{sec:quadratic_cost} and App.~\ref{app:proof}. 
We can combine the findings of the previous section with results from~\cite{CoronGrueneWorthmann20}, where the gap from homogeneous approximations to model predictive control without stabilizing terminal ingredients was recently overcome. 
In~\cite{CoronGrueneWorthmann20}, as usual in MPC without terminal cost or constraints, cost controllability is shown using a growth bound on the value function w.r.t. time using tailored costs. 
Specifically, cost controllability is shown for systems with known, globally asymptotically null controllable homogeneous approximations using robustness of homogeneous control Lyapunov functions~\cite{Gruene00}. 
For that, the approximation's residuum needs to satisfy a certain condition, giving an indication of the approximation quality. 
Based on this, if the system approximation is $(\bm{r},\bm{s},\tau)$-homogeneous, see Sec.~\ref{sec:homog_approx}, the tailored stage cost $\ell:\,\mathbb{R}^{n_x}\times\mathbb{R}^{n_u}\to\mathbb{R}_{\geq 0}$ proposed in~\cite{CoronGrueneWorthmann20} is homogeneous and reads
\begin{align}\label{eq:MPC_cost}
	\ell (\bm{z},\bm{u}) \coloneqq \sum_{i=1}^{n_x} {z}_i^{\frac{d}{r_i}} + \sum_{j=1}^{n_u} {u}_j^{\frac{d}{s_j}}, \quad d\coloneqq 2 \prod_{i=1}^{n_x} r_i.
\end{align} 
Here and in the following, without loss of generality, the system is assumed to be expressed in terms of the privileged coordinates~$\bm{z}$. 
Importantly, it is possible to tune the cost~\eqref{eq:MPC_cost} with positive weighting factors on each state and input, and~$d$ may actually be chosen arbitrarily on the interval~$(0,\infty)$, see~\cite[Rem.\ 4.7]{CoronGrueneWorthmann20}. 
We make use of this generalization in our applications later on. 
In order to conclude that there exists a (sufficiently large) prediction horizon which renders the closed loop subject to the tailored cost~\eqref{eq:MPC_cost} (locally) asymptotically stable, it remains to be shown that~\eqref{eq:system_approx} is indeed a homogeneous approximation of the original system in terms of~\cite[Def.~4.1]{CoronGrueneWorthmann20}. 
For all systems considered in this paper's application, this indeed is the case and can be readily verified. 

For the considered class of systems, the underlying OCP to be solved in each sampling instant reads 
\begin{subequations}
    \begin{alignat}{3}%
            &\hspace{0pt} \underset{\bm{u}(\cdot\,\vert\, t)}{\text{minimize}}
            &&\hspace{6pt}\!\int_{t}^{t+T}\!\ell(\bm{z}(\tau\,\vert\, t),\bm{u}(\tau\,\vert\, t))\; \textnormal{d}\tau \label{eq:mpc_cost_specific}\\
            & \text{subject to}
            &&\hspace{6pt}\dot{\bm{z}}(\tau\,\vert\, t)=\bm{G}(\bm{z}(\tau\,\vert\, t))\bm{u}(\tau\,\vert\, t),\label{eq:mpc_dynamics_specific}\\
            &&&\hspace{6pt} \bm{u}(\tau \,\vert\, t)\in\mathcal{U}\quad \forall \tau \in [t,t+T),\\
            &&&\hspace{6pt}\bm{z}(t\,\vert\, t)=\bm{z}(t),
    \end{alignat}%
    \label{eq:mpc_optprob_specific}%
\end{subequations}%
where the stage cost~$\ell$ follows from~\eqref{eq:MPC_cost} using the privileged coordinates derived in Sec.~\ref{sec:homog_approx}. 
Note that the OCP is subject to the actual dynamics~\eqref{eq:mpc_dynamics_specific} expressed in the privileged coordinates and not the approximate dynamics~\eqref{eq:system_approx}. 
The homogeneous system approximation is therefore only used to derive the cost function, not to predict the system behavior. 

It is now evident for which special cases the authors' original suggestion~\cite{RosenfelderEbelEberhard21} of a direct mechanical interpretation of the cost function, namely that the tailored cost follows from dot products of the state~$\bm{x}$ and the vector fields spanning the Lie algebra evaluated at the setpoint, applies. In case that there is an adapted frame which is an orthogonal basis of $T_{\bm{d}}\mathbb{R}^{n_x}$, it holds that $\bm{A}^{-{\mkern-1.5mu\mathsf{T}}} = \bm{A}$ in the first transformation step~\eqref{eq:Trafo1Bellaiche} of Bella\"{i}che's algorithm. 
Consequently, if, additionally, the second step yields~$\bm{z}=\bm{y}$, which is not uncommon in applications, the privileged coordinates directly follow from the dot product of the vector fields spanning the Lie algebra evaluated at the setpoint and the system's original coordinates~$\bm{x}$. 

Apparently, the utilized optimization problem~\eqref{eq:mpc_optprob_specific}, which (locally) asymptotically stabilizes non-holonomic systems of the form~\eqref{eq:system}, is quite neat and concise, especially compared to other control techniques in the context of non-holonomic systems. In combination with the constructive procedure to derive privileged coordinates, see Sec.~\ref{sec:homog_approx}, this yields a control scheme that is easily applicable to a wide range of non-holonomic systems, as we exemplify in the subsequent section. 

\section{Application to non-holonomic vehicles}
\label{sec:MPC_for_nonhol_vehicles}
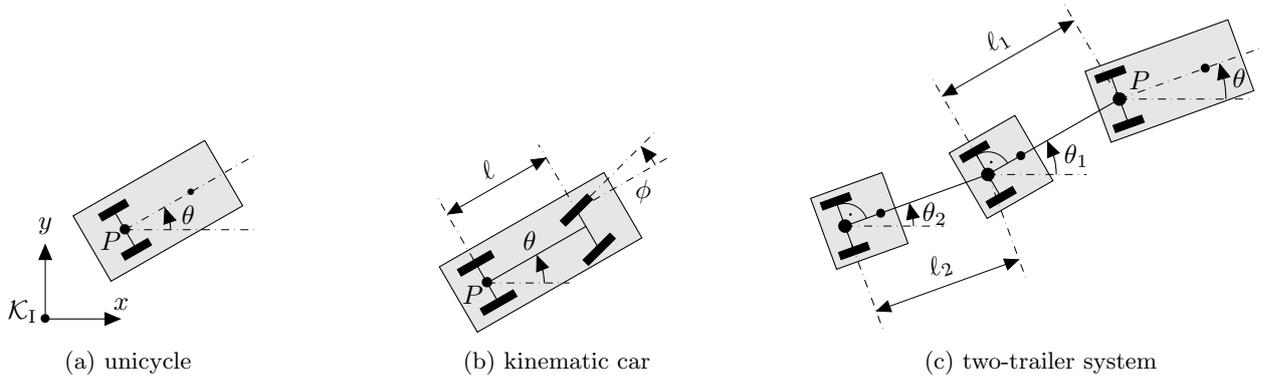
\begin{figure*}[t]
    \centering
	\begin{subfigure}[t]{0.3\linewidth}
		\centering

\begin{tikzpicture}[>= triangle 45]

\def\lC{1}
\def\angle{30}
\def\angleDes{65}
\def\xposR{1}
\def\yposR{-.0}
\def\lengthDD{2}
\def\heightDD{1}

\pgfmathsetmacro\lengthWheel{0.2*\lengthDD}
\pgfmathsetmacro\heightWheel{0.1*\heightDD}
\pgfmathsetmacro\wheelxR{\xposR+0.15*\lengthDD}
\pgfmathsetmacro\wheelxMR{\wheelxR+0.5*\lengthWheel}
\pgfmathsetmacro\wheelLyR{\yposR+(1-0.15)*\heightDD-\heightWheel}
\pgfmathsetmacro\wheelRyR{\yposR+0.15*\heightDD}

\pgfmathsetmacro\wheelMyR{\yposR+0.5*\heightDD}
\pgfmathsetmacro\wheelxMRangle{\wheelxMR+\lengthDD}
\pgfmathsetmacro\wheelMyRhor{\wheelMyR-cos(\angle)*sin(\angle)*\lengthDD}
\pgfmathsetmacro\wheelxMRangleHor{\wheelxMR+cos(\angle)*cos(\angle)*\lengthDD}
\pgfmathsetmacro\radAngle{0.3*\lengthDD}
\pgfmathsetmacro\drawAngleX{\wheelxMR+\radAngle*cos(\angle)}
\pgfmathsetmacro\drawAngleY{\wheelMyR-\radAngle*sin(\angle)}

\pgfmathsetmacro\wheelFrontX{\wheelxMR+0.5*\lengthDD}
\pgfmathsetmacro\wheelFrontY{\wheelMyR}

\pgfmathsetmacro\wheelxMRa{\wheelxMR-1.65*\lengthWheel}

\node [black] at (0,0) {\footnotesize \textbullet};
	\node at (-0.3,0.1) {$\mathcal{K}_{\text{I}}$};
\draw[->, thin]  (0,0) -- (\lC,0) node[above, xshift = 0cm, yshift = -0cm] {$x$};
\draw[->, thin]  (0,0) -- (0,\lC) node[above, xshift = 0cm, yshift = -0cm] {$y$};

\begin{scope}[rotate = \angle]
\draw[draw = black, fill = grayCar] (\xposR,\yposR) rectangle ++(\lengthDD,\heightDD);
\draw[draw = black, fill = black] (\wheelxR,\wheelLyR) rectangle ++(\lengthWheel,\heightWheel);
\draw[draw = black, fill = black] (\wheelxR,\wheelRyR) rectangle ++(\lengthWheel,\heightWheel);
\draw[] (\wheelxMR,\wheelRyR) -- (\wheelxMR,\wheelLyR);
\draw[very thin, dashdotted] (\wheelxMR,\wheelMyR) -- (\wheelxMRangle,\wheelMyR);
\draw[very thin, dashdotted] (\wheelxMR,\wheelMyR) -- (\wheelxMRangleHor,\wheelMyRhor);
\draw[thin, ->] (\drawAngleX,\drawAngleY) arc(-\angle:0:\radAngle) node[right, xshift = 0.1cm, yshift = -0.1cm] {$\theta$};
\node [] at (\wheelFrontX,\wheelFrontY) {\tiny \textbullet};
\node [] at (\wheelxMR,\wheelMyR) {\textbullet};
\node [right, xshift = -0.45cm, yshift = -0.15cm] at (\wheelxMR,\wheelMyR) {$P$};

\end{scope}


%

\end{tikzpicture}

		\caption[]%
		{{unicycle}}
		\label{fig:sketch_unicycle}
	\end{subfigure}
	\hfill
    \begin{subfigure}[t]{0.3\linewidth}
		\centering
		\definecolor{coordinates}{RGB}{0, 0, 0}

\begin{tikzpicture}[>= triangle 45]

\def\lC{1}
\def\angle{30}
\def\angleDes{65}
\def\angleFront{15}
\def\xposR{0}
\def\yposR{0}
\def\lengthDD{2.5}
\def\heightDD{1}

\pgfmathsetmacro\lengthWheel{0.2*\lengthDD}
\pgfmathsetmacro\heightWheel{0.1*\heightDD}
\pgfmathsetmacro\wheelxR{\xposR+0.075*\lengthDD}
\pgfmathsetmacro\wheelxMR{\wheelxR+0.5*\lengthWheel}
\pgfmathsetmacro\wheelLyR{\yposR+(1-0.15)*\heightDD-\heightWheel}
\pgfmathsetmacro\wheelRyR{\yposR+0.15*\heightDD}

\pgfmathsetmacro\wheelMyR{\yposR+0.5*\heightDD}
\pgfmathsetmacro\wheelMyRL{\wheelLyR}
\pgfmathsetmacro\wheelMyRR{\wheelRyR}
\pgfmathsetmacro\wheelxMRangle{\wheelxMR+\lengthDD + 0.5}

\pgfmathsetmacro\wheelMyRhor{\wheelMyR-cos(\angle)*sin(\angle)*\lengthDD*0.5}
\pgfmathsetmacro\wheelxMRangleHor{\wheelxMR+cos(\angle)*cos(\angle)*\lengthDD*0.5}
\pgfmathsetmacro\radAngle{0.3*\lengthDD}

\pgfmathsetmacro\drawAngleX{\wheelxMR+\radAngle*cos(\angle)}
\pgfmathsetmacro\drawAngleY{\wheelMyR-\radAngle*sin(\angle)}

\pgfmathsetmacro\wheelFrontX{\wheelxMR+0.5*\lengthDD}
\pgfmathsetmacro\wheelFrontYL{\wheelMyRL-0.5*\heightWheel}
\pgfmathsetmacro\wheelFrontYR{\wheelMyRR-0.5*\heightWheel}
\pgfmathsetmacro\wheelFrontXC{\wheelFrontX+0.5*\lengthWheel}
\pgfmathsetmacro\radPhi{1.2}
\pgfmathsetmacro\drawPhiX{\wheelFrontXC+\radPhi}
\pgfmathsetmacro\drawPhiY{\wheelMyRL}

\pgfmathsetmacro\wheelxMRa{\wheelxMR-1.65*\lengthWheel}
\pgfmathsetmacro\wheelLyRb{\wheelLyR+1*\heightWheel}


\begin{scope}[rotate = \angle]
\draw[draw = black, fill = grayCar] (\xposR,\yposR) rectangle ++(\lengthDD,\heightDD);
\draw[draw = black, fill = black] (\wheelxR,\wheelLyR) rectangle ++(\lengthWheel,\heightWheel);
\draw[draw = black, fill = black] (\wheelxR,\wheelRyR) rectangle ++(\lengthWheel,\heightWheel);
\draw[] (\wheelxMR,\wheelRyR) -- (\wheelxMR,\wheelLyR);
\draw[very thin, dashdotted, draw = coordinates] (\wheelxMR,\wheelMyR) -- (\wheelxMRangleHor,\wheelMyRhor);
\draw[thin, ->, draw = coordinates] (\drawAngleX,\drawAngleY) arc(-\angle:0:\radAngle) node[right, xshift = -0.3cm, yshift = 0.185cm, coordinates] {$\theta$};
\draw[] (\wheelxMR,\wheelMyR) -- (\wheelFrontXC,\wheelMyR);
\draw[draw = black, fill = black, rotate around={\angleFront:(\wheelFrontXC,\wheelMyRL)}] (\wheelFrontX,\wheelFrontYL) rectangle ++(\lengthWheel,\heightWheel);
\draw[draw = black, fill = black, rotate around={\angleFront:(\wheelFrontXC,\wheelMyRR)}] (\wheelFrontX,\wheelFrontYR) rectangle ++(\lengthWheel,\heightWheel);
\draw[] (\wheelFrontXC,\wheelMyRL) -- (\wheelFrontXC,\wheelMyRR);
\node [coordinates] at (\wheelxMR,\wheelMyR) {\textbullet};
\node [right, xshift = -0.45cm, yshift = -0.15cm, coordinates] at (\wheelxMR,\wheelMyR) {$P$};

\def\lPhi{1.5}
\pgfmathsetmacro\xPhi{\wheelFrontXC+\lPhi*cos(\angleFront)}
\pgfmathsetmacro\yPhi{\wheelMyRL+\lPhi*sin(\angleFront)}
\draw[very thin, dashdotted] (\wheelFrontXC,\wheelMyRL) -- (\xPhi,\yPhi);
\draw[thin, ->, draw = coordinates] (\drawPhiX,\drawPhiY) arc(0:\angleFront:\radPhi) node[right, xshift = -0.2cm, yshift = -0.6cm, coordinates] {$\phi$};
\draw[very thin, dashdotted, draw = coordinates] (\wheelFrontXC,\wheelMyRL) -- ++(1.4,0);

\draw[very thin, dashdotted] (\wheelxMR,\wheelLyRb) -- ++(0,1);
\draw[very thin, dashdotted] (\wheelFrontXC,\wheelFrontYL) -- ++(0,1.05);
\draw[very thin, <->] (\wheelxMR,\wheelLyRb+0.7) -- ++(1.5,0) node[midway, above, rotate = \angle] { $\ell$};

\end{scope}

\end{tikzpicture}
		\caption[]%
		{{kinematic car}}
		\label{fig:sketch_kinCar}
	\end{subfigure}
    \hfill
    \begin{subfigure}[t]{0.38\linewidth}
		\centering
		\definecolor{coordinates}{RGB}{0, 0, 0}

\begin{tikzpicture}[>= triangle 45]

\def\lC{1}
\def\lTrailer{2}
\def\angle{20}
\def\angleT{10}
\def\angleTB{-10}
\def\angleDes{65}
\def\xpos{9.5}
\def\ypos{1.5}
\def\lengthDD{2}
\def\heightDD{1}
\def\lengthT{1}
\def\heightT{\heightDD}

\pgfmathsetmacro\offsetP{0.15*\lengthDD}
\pgfmathsetmacro\lengthWheel{0.2*\lengthDD}
\pgfmathsetmacro\heightWheel{0.1*\heightDD}

\pgfmathsetmacro\wheelxR{-0.5*\lengthWheel}
\pgfmathsetmacro\wheelxMR{\wheelxR+0.5*\lengthWheel}
\pgfmathsetmacro\wheelLyR{0.9*\heightDD-\heightWheel - 0.5*\heightDD}
\pgfmathsetmacro\wheelRyR{0.1*\heightDD- 0.5*\heightDD}

\pgfmathsetmacro\wheelMyR{\ypos+0.5*\heightDD}
\pgfmathsetmacro\wheelxMRangle{\wheelxMR+\lengthDD}
\pgfmathsetmacro\radAngle{0.7*\lengthDD}
\pgfmathsetmacro\drawAngleX{\radAngle*cos(\angle)}
\pgfmathsetmacro\drawAngleY{-\radAngle*sin(\angle)}
\pgfmathsetmacro\wheelFrontX{\wheelxMR+0.6*\lengthDD}

	\pgfmathsetmacro\wheelFrontT{\wheelxMR+0.5*\lengthT}
	\pgfmathsetmacro\radAngleT{0.9*\lengthT}
	\pgfmathsetmacro\drawAnglexT{\radAngleT*cos(\angleT+\angle)}
	\pgfmathsetmacro\drawAngleyT{-\radAngleT*sin(\angleT+\angle)}
	
	\pgfmathsetmacro\drawAnglexTB{\radAngleT*cos(\angleTB+\angleT+\angle)}
	\pgfmathsetmacro\drawAngleyTB{-\radAngleT*sin(\angleTB+\angleT+\angle)}


\begin{scope}[xshift = \xpos cm, yshift = \ypos cm, rotate = \angle]
	\draw[draw = black, fill = grayCar] (-\offsetP, -0.5*\heightDD) rectangle ++(\lengthDD,\heightDD);
	\draw[draw = black, fill = black] (\wheelxR,\wheelLyR) rectangle ++(\lengthWheel,\heightWheel);
	\draw[draw = black, fill = black] (\wheelxR,\wheelRyR) rectangle ++(\lengthWheel,\heightWheel);
	\draw[] (\wheelxMR,\wheelRyR) -- (\wheelxMR,\wheelLyR);
	\draw[very thin, dashdotted] (0,0) -- ++(\lengthDD,0);
	\draw[thin, ->, draw = coordinates] (\drawAngleX,\drawAngleY) arc(-\angle:0:\radAngle) node[right, xshift = 0.05cm, yshift = -0.2cm, color = coordinates] {$\theta$};
	\node [] at (\wheelFrontX,0)[circle,fill,inner sep=1.2pt]{};
	\node [color = coordinates] at (0,0)[circle,fill,inner sep=1.8pt]{};
	\node [right, xshift = -0.0cm, yshift = 0.25cm, color = coordinates] at (\wheelxMR,0) {$P$};
	\coordinate (originDIANA) at (0,0);

	\begin{scope}[rotate = \angleT, xshift = -\lTrailer cm]
		\coordinate (originT) at (0,0);
		\draw[draw = black, fill = grayCar] (-\offsetP, -0.5*\heightT) rectangle ++(\lengthT,\heightT);
		\node [] at (0,0)[circle,fill,inner sep=1.8pt]{};
		\draw[draw = black, fill = black] (\wheelxR,\wheelLyR) rectangle ++(\lengthWheel,\heightWheel);
		\draw[draw = black, fill = black] (\wheelxR,\wheelRyR) rectangle ++(\lengthWheel,\heightWheel);
		\draw[] (\wheelxMR,\wheelRyR) -- (\wheelxMR,\wheelLyR);
		\draw[thin, ->, draw = coordinates] (\drawAnglexT,\drawAngleyT) arc(-\angleT-\angle:0:\radAngleT)
		 node[right, xshift = 0.1cm, yshift = -0.2cm, color = coordinates] {$\theta_1$};
		\node [] at (\wheelFrontT,0)[circle,fill,inner sep=1.2pt]{};
		\draw[] (0,0) -- (originDIANA);
		\draw[very thin] (0.3,0) arc(0:90:0.3);
		\node [] at (0.12,0.12)[circle,fill,inner sep=0.2pt]{};
		\def\heightNote{1.5}
		\draw[very thin, dashdotted] (0,0) -- ++(0, \heightNote);
		\draw[very thin, dashdotted] (\lTrailer,0) -- ++(0, \heightNote);
		\draw[very thin, <->] (0,0.8*\heightNote) -- ++(\lTrailer, 0) node[midway, above, rotate = \angle+\angleT] {$\ell_1$};
		
		\begin{scope}[rotate = \angleTB, xshift = -\lTrailer cm]
			\coordinate (originTB) at (0,0);
			\draw[draw = black, fill = grayCar] (-\offsetP, -0.5*\heightT) rectangle ++(\lengthT,\heightT);
			\node [] at (0,0)[circle,fill,inner sep=1.8pt]{};
			\draw[draw = black, fill = black] (\wheelxR,\wheelLyR) rectangle ++(\lengthWheel,\heightWheel);
			\draw[draw = black, fill = black] (\wheelxR,\wheelRyR) rectangle ++(\lengthWheel,\heightWheel);
			\draw[] (\wheelxMR,\wheelRyR) -- (\wheelxMR,\wheelLyR);
			\draw[thin, ->, draw = coordinates] (\drawAnglexTB,\drawAngleyTB) arc(-\angleTB-\angleT-\angle:0:\radAngleT) node[right, xshift = 0.05cm, yshift = -0.15cm, color = coordinates] {$\theta_2$};
			\node [] at (\wheelFrontT,0)[circle,fill,inner sep=1.2pt]{};
			\draw[] (0,0) -- (originT);
			\draw[very thin] (0.3,0) arc(0:90:0.3);
			\node [] at (0.12,0.12)[circle,fill,inner sep=0.2pt]{};
			\def\heightNoteB{-1.5}
			\draw[very thin, dashdotted] (0,0) -- ++(0, \heightNoteB);
			\draw[very thin, dashdotted] (\lTrailer,0) -- ++(0, \heightNoteB);
			\draw[very thin, <->] (0,0.8*\heightNoteB) -- ++(\lTrailer, 0) node[midway, above, rotate = \angle+\angleT+\angleTB] {$\ell_2$};
			\end{scope}

	\end{scope}

\end{scope}

\draw[very thin, dashdotted, draw = coordinates] (originTB) -- ++(0.65*\lengthDD,0);

\draw[very thin, dashdotted, draw = coordinates] (originT) -- ++(0.65*\lengthDD,0);

\draw[very thin, dashdotted, draw = coordinates] (originDIANA) -- ++(0.85*\lengthDD,0);

\end{tikzpicture}
		\caption[]%
		{{two-trailer system}}
		\label{fig:sketch_trailer2}
	\end{subfigure}
\caption{Considered non-holonomic vehicles.}
\label{fig:sketch_vehicles}
\end{figure*}
The introduced procedure to suitably approximate driftless, input-affine systems can be deployed for many applications. 
One exemplary case  are robots subject to non-integrable Pfaffian constraints. 
For these, the input matrix~$\bm{G}$ and, thus, the dynamics of system~\eqref{eq:system} directly follow from the kernel of the matrix describing these constraints~\cite{Murray2017}. 
Non-integrable constraints arise, e.g., in wheeled mobile robotics under the commonly used assumption that wheels roll without slipping. 
Hence, following this assumption, this section applies the proposed MPC design methodology for the wheeled robotic systems schematically illustrated in Fig.~\ref{fig:sketch_vehicles}. 
Namely, at first, the widely popular differentially driven mobile robot is considered. 
It is actuated by two independently driven wheels mounted on a common axis, using an additional caster wheel to prevent the robot from tipping over, see Fig.~\ref{fig:sketch_vehicles}(\subref{fig:sketch_unicycle}). 
Due to kinematic equivalence, it is often referred to as a unicycle, which, for brevity, is the term used subsequently. 
The natural extension of the unicycle is the so-called kinematic car or car-like robot~\cite{DeLuca1998}. 
This robot can turn its front axle and is either front-wheel or rear-wheel driven, see Fig.~\ref{fig:sketch_vehicles}(\subref{fig:sketch_kinCar}). However, we consider the front-wheel driven kinematic car in the following. 
In addition, unicycle-trailer combinations are considered with one and two trailers, see Fig.~\ref{fig:sketch_vehicles}(\subref{fig:sketch_trailer2}).
Before deriving tailored non-quadratic cost for these robots, it is worth recalling that quadratic costs are provably insufficient for all of these vehicles, see App.~\ref{app:proof}. 

\subsection{Differentially driven mobile robot (unicycle)}\label{sec:unicycle}
The robot's kinematics can be expressed by means of the model
\begin{align}\label{eq:unicycle}
    \dot{\bm{x}} = \begin{bmatrix}
        \dot{x} \\ \dot{y} \\ \dot{\theta} 
    \end{bmatrix} = \begin{bmatrix}
        \cos \theta & 0 \\ \sin \theta & 0 \\ 0 & 1
    \end{bmatrix} \begin{bmatrix}
        v \\ \omega 
    \end{bmatrix} = \begin{bmatrix}
        \bm{X}_1 (\bm{x}) & \bm{X}_2 (\bm{x})
    \end{bmatrix} \bm{u}, 
\end{align} 
where the state $\bm{x}\in\mathbb{R}^{n_x},~n_x=3$, consists of the spatial position of the point $P$ in the inertial frame of reference~$\mathcal{K}_{\text{I}}$ and the robot's orientation measured relative to the positive $x$-axis, see Fig.~\refeq{fig:sketch_vehicles}(\subref{fig:sketch_unicycle}). Its control input $\bm{u}$ comprises the translational velocity~$v$ and rotational velocity~$\omega$. Naturally, these control inputs, and all of the subsequently considered ones, are constrained, i.e., it holds that $\bm{u}\in\mathcal{U}\subset\mathbb{R}^{n_u},~n_u=2$.
Under the assumption that the robot's wheels roll without slipping, the unicycle~\eqref{eq:unicycle} is subject to the kinematic constraint $\dot{x}\sin \theta -\dot{y}\cos \theta =0$, which enforces that the lateral velocity is zero. 
However, this constraint is non-integrable since the unicycle can be steered to any setpoint in the state space using the implicit control direction following from the Lie bracket $\bm{X}_3 \coloneqq [\bm{X}_1, \, \bm{X}_2]$. This yields the distributions $\bm{\Delta}^1 = \text{span} \{ \bm{X}_1, \, \bm{X}_2 \}$ and $\bm{\Delta}^2 = \text{span} \{ \bm{X}_1, \, \bm{X}_2, \, \bm{X}_3 \}$. 
Consequently, controllability can be shown, e.g., using the LARC and further, the growth vector follows as $\bm{n} = (2, \, 3)$ for an arbitrary setpoint $\bm{d}\in\mathbb{R}^{n_x}$. This in turn yields the weights $\bm{w} = (1,\,1,\,2)$ such that $\{ \bm{X}_1(\bm{d}), \, \bm{X}_2(\bm{d}), \, \bm{X}_3 (\bm{d}) \}$ indeed forms an adapted frame at $\bm{d}$. 
Next, in order to check whether $\bm{x}$ is already a set of privileged coordinates, the non-holonomic orders of the coordinate functions $x_i,~i\in\mathbb{Z}_{1:n_x}$, are derived. 
To this end, the non-holonomic derivatives $\bm{X}_1 x_1 (\bm{d})  = \cos d_3$, $\bm{X}_1 x_2 (\bm{d}) = \sin d_3$, $\bm{X}_2 x_3 (\bm{d}) = 1$, $\bm{X}_2(\bm{X}_1 x_1)(\bm{d})=-\sin d_3$, and $\bm{X}_2(\bm{X}_1 x_2)(\bm{d})=\cos d_3$ are considered. 
Thus the orders read $\text{ord}_{\bm{d}}(x_3) = 1$ for any $\bm{d}\in\mathbb{R}^{n_x}$, $\text{ord}_{\bm{d}} (x_1)= 2$ if $d_3 = (2k-1)\pi/2,~k\in\mathbb{Z}$, and $\text{ord}_{\bm{d}} (x_1)= 1$ otherwise, as well as $\text{ord}_{\bm{d}} (x_2)= 2$ if $d_3 = k\pi,~k\in\mathbb{Z}$, and $\text{ord}_{\bm{d}} (x_2)= 1$ otherwise. 
This implies that reordering $\bm{x}$ suffices for setpoints with $d_3=k\pi/2,~k\in\mathbb{Z}$, to obtain privileged coordinates since the necessary condition $\text{ord}_{\bm{d}} (x_j) = w_j,~j\in\mathbb{Z}_{1:n_x}$, is met. 
However, for general setpoints $\bm{d}\in\mathbb{R}^3$, the condition does not hold and Bella\"{i}che's algorithm is employed to obtain the privileged coordinates.

Utilizing the previously determined adapted frame, the original coordinates $\bm{x}$ are transformed in the first step~\eqref{eq:Trafo1Bellaiche} of Bella\"{i}che's algorithm to
\begin{align}\label{eq:trafo_unicycle}
	\bm{y} 
    &= \begin{bmatrix}
        (x_1-d_1) \cos d_3 + (x_2-d_2) \sin d_3 \\ x_3 - d_3 \\ (x_1-d_1) \sin d_3 - (x_2-d_2) \cos d_3
    \end{bmatrix}.
\end{align} Since it holds that~$w_j < 3$, it immediately follows from~\eqref{eq:Trafo2Bellaiche} that the privileged coordinates read $\bm{z}=\bm{y}$. Illustratively, these coordinates describe the unicycle's position in a rotated frame of reference  having its $x$-axis pointing along the desired orientation $d_3$. Note that the resulting privileged coordinates~\eqref{eq:trafo_unicycle} match the derivation from the mechanical point of view presented in~\cite{RosenfelderEbelEberhard21,RosenfelderEbelEberhard22}. Finally, the homogeneous nilpotent approximation can be obtained by expanding the control vector fields for the system described in the privileged coordinates, i.e.,   
\begin{align*}
    \bm{Z}_1 (\bm{z}) = \begin{bmatrix}
        \cos z_2 & 0 & -\sin z_2
    \end{bmatrix}\tran, \quad \bm{Z}_2 (\bm{z}) = \begin{bmatrix}
        0 & 1 & 0
    \end{bmatrix}\tran.
\end{align*} The corresponding expansions can be written as
\begin{align*}
    \bm{Z}_1 (\bm{z}) ={}& \left( \sum_{k=0}^\infty (-1)^k  \frac{z_2^{2k}}{(2k)!}  \right)\! \frac{\partial}{\partial z_1}  - \left( \sum_{k=0}^\infty \frac{z_2 ^{2k+1}}{(2k+1)!} \right)\! \frac{\partial}{\partial z_3} 
\end{align*} 
and $\bm{Z}_2 (\bm{z}) = \frac{\partial}{\partial z_2}$. 
Then, the weighted degree $w(\bm{\alpha})$ of each monomial $\bm{z}^{\bm{\alpha}} $ is determined which yields the weighted degree $w(\bm{\alpha})-w_j$ of the corresponding monomial vector field $\bm{z}^{\bm{\alpha}}\frac{\partial}{\partial z_j}$. 
Exemplarily, the first monomial $z_2$ of the sine's expansion implies $\bm{\alpha}=(0,\,1,\,0)$ which in turn yields $w(\bm{\alpha}) - w_3 = 1\cdot w_2-w_3 = -1$ for the monomial vector field~$z_2\frac{\partial}{\partial z_3}$, see Sec.~\ref{sec:homog_approx}. 
Grouping together all crucial monomial vector fields, i.e., those with weighted degree $-1$, we obtain $\bm{Z}_1^{[-1]}=[1 \ 0 \ -z_2]\tran$ and $\bm{Z}_2^{[-1]}=[0 \ 1 \ 0]\tran$. 
These control vector fields build the nilpotent approximation~\eqref{eq:system_approx} for the unicycle~\eqref{eq:unicycle} under the transformation $\bm{z}=\bm{y}$, with $\bm{y}$ from~\eqref{eq:trafo_unicycle}, at any point $\bm{d}$. Moreover, this approximation is homogeneous with $\tau=0$, $\bm{s}=(1,\,1)$, and $\bm{r}=(1,\,1,\,2)$. 
Since the resulting homogeneous system is indeed a suitable approximation, see~\cite[Ex.~4.3]{CoronGrueneWorthmann20}, we can employ it to obtain the tailored OCP~\eqref{eq:mpc_optprob_specific} with the stage cost 
\begin{align}\label{eq:cost_unicycle}
    \ell (\bm{x},\bm{u})\! ={}&q_1 \left( (x_1 - d_1) \cos d_3 + (x_2-d_2) \sin d_3 \right)^4 + q_2 x_3^4  \nonumber \\{}&+ q_3 \left( (x_1-d_1) \sin d_3 - (x_2-d_2) \cos d_3  \right)^2 \nonumber\\  {}&+ r_1 v^4 + r_2 \omega^4
\end{align} 
Again, we refer to~\cite{RosenfelderEbelEberhard21,RosenfelderEbelEberhard22} for an alternative derivation from a practical perspective yielding the same result. 
Moreover, note that this universal description degenerates to the unicycles's approximation at the origin $\bm{d}=\bm{0}$ given in~\cite{CoronGrueneWorthmann20}. 
Our previous work~\cite{RosenfelderEbelEberhard21,RosenfelderEbelEberhard22} already contains simulative and experimental results for the unicycle using the cost~\eqref{eq:cost_unicycle}. 
Thus, for brevity, we refrain from reprinting extensive results here and merely summarize that the resulting controller works as expected.
This is exemplarily depicted by means of a forward parking scenario in Fig.~\ref{fig:DIANA_experiment}. 
There, the unicycle shall drive from the origin to the desired setpoint highlighted by a red cross with a goal orientation of $\unit[\pi/4]{rad}$ relative to the horizontal image axis. 
Subsequently, we focus on how the novel aspects from the previous sections enable us to deal with more intricate mobile robots by applying the same constructive procedure as just performed for the unicycle. 

\begin{figure}[t]
	\centering
    \begin{subfigure}[t]{0.6\linewidth}
        \def\SVGwidthC{0.95\textwidth}
		\def\svgwidth{\SVGwidthC}
\begingroup%
  \makeatletter%
  \providecommand\color[2][]{%
    \errmessage{(Inkscape) Color is used for the text in Inkscape, but the package 'color.sty' is not loaded}%
    \renewcommand\color[2][]{}%
  }%
  \providecommand\transparent[1]{%
    \errmessage{(Inkscape) Transparency is used (non-zero) for the text in Inkscape, but the package 'transparent.sty' is not loaded}%
    \renewcommand\transparent[1]{}%
  }%
  \providecommand\rotatebox[2]{#2}%
  \newcommand*\fsize{\dimexpr\f@size pt\relax}%
  \newcommand*\lineheight[1]{\fontsize{\fsize}{#1\fsize}\selectfont}%
  \ifx\svgwidth\undefined%
    \setlength{\unitlength}{435.53791809bp}%
    \ifx\svgscale\undefined%
      \relax%
    \else%
      \setlength{\unitlength}{\unitlength * \real{\svgscale}}%
    \fi%
  \else%
    \setlength{\unitlength}{\svgwidth}%
  \fi%
  \global\let\svgwidth\undefined%
  \global\let\svgscale\undefined%
  \makeatother%
  \begin{picture}(1,0.74973564)%
    \lineheight{1}%
    \setlength\tabcolsep{0pt}%
    \put(0,0){\includegraphics[width=\unitlength,page=1]{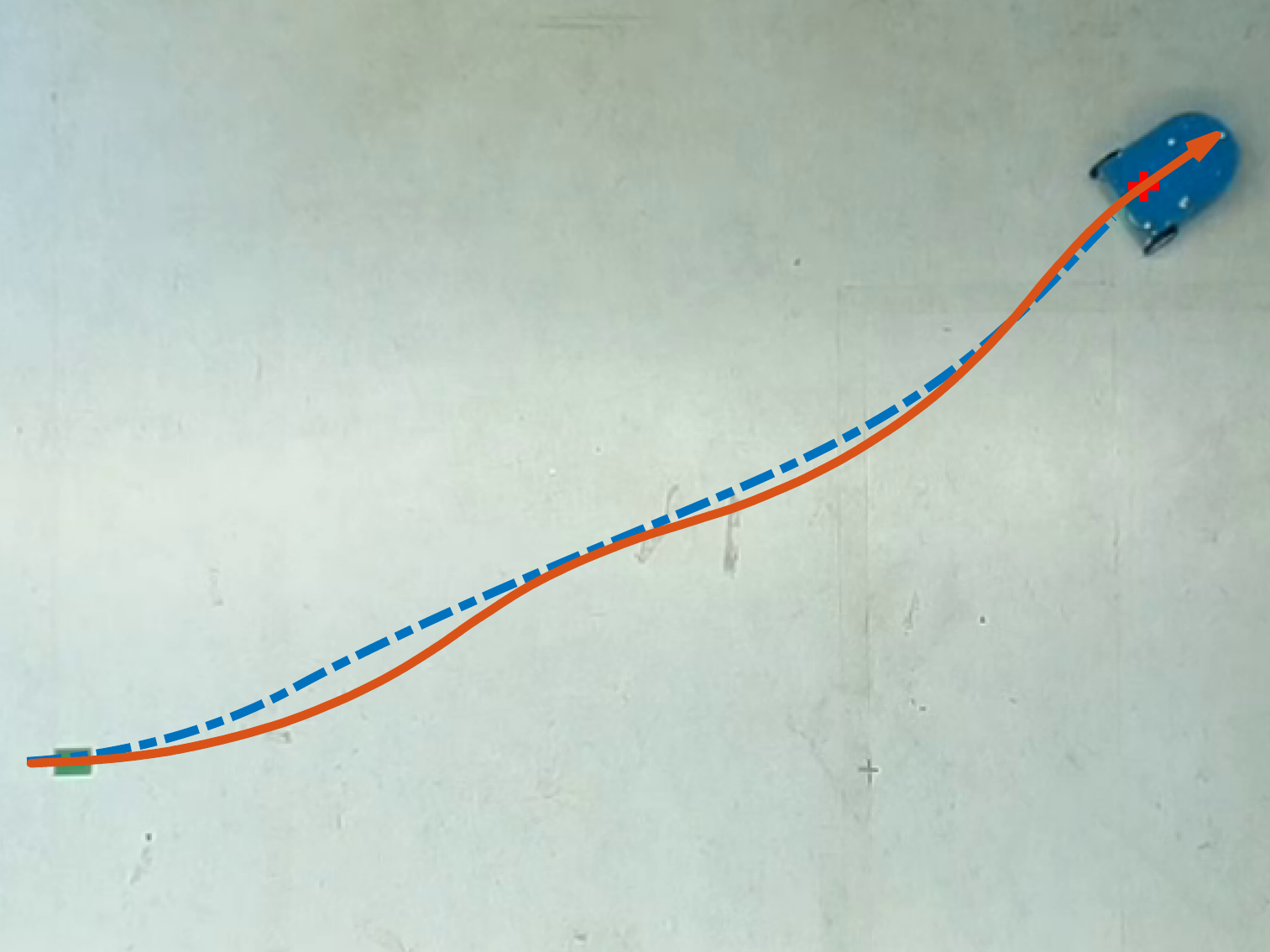}}%
    \put(0.0118159,0.0114315){\makebox(0,0)[lt]{\lineheight{1.25}\smash{\begin{tabular}[t]{l}$t = 6.00\,\textnormal{s}$\end{tabular}}}}%
  \end{picture}%
\endgroup%

    \end{subfigure}
    \caption{Comparison of the experimental (orange line)  and simulative (blue dash-dotted line) result for the unicycle.}
    \label{fig:DIANA_experiment} 
\end{figure}%

\textbf{Remark} For any non-holonomic system of the form~\eqref{eq:system} for which it holds $w_{n_x}<3$, as in this example, the second step~\eqref{eq:Trafo2Bellaiche} of Bella\"{i}che's algorithm can be skipped since then it holds $\bm{z}=\bm{y}$. 
In particular, this is the case for systems with a degree of non-holonomy $r<3$ since~$w_{n_x}=r$. 
This considerably reduces the algorithm's complexity.

\subsection{Kinematic car}
A non-holonomic vehicle that is of particular interest due to its kinematic resemblance to an automobile is the planar so-called kinematic car, or car-like robot. 
Moreover, from a control-theoretic perspective, the kinematic car is more intricate to control than the unicycle since the car's maneuverability is even more restricted. 
As will be seen later, this is also tangible by means of the increased degree of non-holonomy. 
The kinematics of the front-wheel driven kinematic car is described by
\begin{align}\label{eq:kinematic_car}
    \dot{\bm{x}} = \begin{bmatrix}
        \dot{x} \\ \dot{y} \\ \dot{\theta} \\ \dot{\phi}
    \end{bmatrix} = \begin{bmatrix}
        \cos \theta \cos \phi & 0 \\
        \sin \theta \cos \phi & 0 \\
        \frac{1}{\ell} \sin \phi & 0 \\
        0 & 1
    \end{bmatrix} \begin{bmatrix}
        v \\ \omega
    \end{bmatrix} = \begin{bmatrix}
        \bm{X}_1 & \bm{X}_2
    \end{bmatrix} \bm{u},
\end{align} 
where the first three states are the same as for the unicycle and the additional fourth state is the steering angle of the front axis, see also Fig.~\ref{fig:sketch_vehicles}(\subref{fig:sketch_kinCar}). 
The length $\ell>0$ describes the distance between the two axles. 
The car is actuated by means of the front-wheels' translational velocity~$v$ and the angular steering velocity~$\omega$. 
System~\eqref{eq:kinematic_car} is subject to the same kinematic constraint as the unicycle which describes the restricted movement capabilities of the rear axle. 
Analogously, a second non-holonomic constraint $\dot{x} \sin (\theta+\phi) - \dot{y} \cos (\theta + \phi) - \dot{\theta} \ell \cos \phi = 0$ describes the kinematic constraint of the front axle. 
The system's controllability can be shown, e.g., using the LARC with the Lie brackets 
\begin{align}
    \bm{X}_3  &= [\bm{X}_1, \, \bm{X}_2] = \begin{bmatrix}
        \cos \theta \sin \phi & \sin\theta \sin \phi & - \frac{1}{\ell} \cos \phi & 0
    \end{bmatrix}\tran,  \nonumber\\
    \bm{X}_4  &= [\bm{X}_1, \, \bm{X}_3] = \begin{bmatrix}
        -\frac{1}{\ell} \sin \theta & \frac{1}{\ell} \cos \theta & 0 & 0
    \end{bmatrix}\tran.
\end{align}
For more insight regarding the kinematic car, especially the rear-wheel driven case, we refer to~\cite{DeLuca1998}. 
An analogous cost function for the rear-wheel driven car is given in~\cite{RosenfelderEbelEberhard21}. 
However, the derivation presented therein does not follow this paper's constructive and transferable procedure.

Choosing the first two distributions as for the unicycle as well as $\bm{\Delta}^3 = \text{span} \{\bm{X}_1, \, \bm{X}_2, \, \bm{X}_3, \, \bm{X}_4 \}$, the growth vector reads $\bm{n}=(2,\, 3, \, 4)$ for any $\bm{d}\in\mathbb{R}^{n_x},~n_x = 4$. 
Consequently, the weights read $\bm{w} = (1, \, 1,\,2,\,3)$ such that the span of all $\bm{X}_j,~j\in\mathbb{Z}_{1:n_x}$, forms an adapted frame at the setpoint $\bm{d}$. 
We aim to drive the car to the origin and, hence, we set $\bm{d}=\bm{0}$ subsequently. 
This yields the first-order non-holonomic derivatives of the coordinate functions $\bm{X}_1 x_1 (\bm{0}) = 1$, $\bm{X}_1 x_j(\bm{0})=0$ for $j\in\mathbb{Z}_{2:3}$, and $\bm{X}_2 x_4 (\bm{0})=1$. 
Further, it holds that $\bm{X}_2 \bm{X}_1 x_3 (\bm{0})=1/\ell$ and $\bm{X}_2 \bm{X}_1^2 x_2 (\bm{0})=1/\ell$, whereas $\bm{X}_2\bm{X}_1 x_2(\bm{0})=\bm{X}_1^2x_2(\bm{0})=0$. 
Hence, the kinematic car's non-holonomic orders at the origin read $\text{ord}_{\bm{0}} (x_1) = 1$, $\text{ord}_{\bm{0}} (x_2) = 3$, $\text{ord}_{\bm{0}} (x_3) = 2$, and $\text{ord}_{\bm{0}} (x_4) = 1$. 
Thus, reordering the coordinates $\bm{x}$ suffices to obtain a system of privileged coordinates, cf. Sec.~\ref{sec:homog_approx}. However, for the sake of completeness we continue with Bella\"{i}che's algorithm since it will end up in a canonical form~\cite{Bellaiche1993}.

The first step~\eqref{eq:Trafo1Bellaiche} of Bella\"{i}che's algorithm is given as
\begin{align}\label{eq:trafo_car}
	\bm{y} &= \left.\begin{bmatrix} \bm{X}_1  & \bm{X}_2  & \bm{X}_{3} & \bm{X}_{4}  \end{bmatrix}\right|_{\bm{0}}^{-1} {\bm{x}} \nonumber\\ {}&= \begin{bmatrix}
        x_1 & x_4 &-\ell x_3 & \ell x_2
    \end{bmatrix}\tran.
\end{align}
In the second step~\eqref{eq:Trafo2Bellaiche}, the weights $\bm{w}(\bm{0})$ yield the coordinates $\bm{z}$ in terms of $z_j= y_j,~j\in\mathbb{Z}_{1:3}$, and $z_4=y_4-h_{4,2}(y_1,y_2,y_3)$. 
The addend given therein reads
\begin{align*}
    h_{4,2} (y_1,y_2,y_3) = \sum_{\substack{|\bm{\alpha}|=2 \\ w(\bm{\alpha})< w_4}} \bm{Y}_1^{\alpha_1} \left. \bm{Y}_2^{\alpha_2} \bm{Y}_3^{\alpha_3} y_4 \right|_{\bm{0}} \prod_{i=1}^3 \frac{y_i^{\alpha_i}}{\alpha_i !},
\end{align*} where the sum's conditions hold for $\bm{\alpha}=(1,\,1,\,0,\,0)$, $\bm{\alpha}=(2,\,0,\,0,\,0)$, and $\bm{\alpha}=(0,\,2,\,0,\,0)$ due to $w_4 = 3$. However, for setpoints with $d_4 = 0$, all relevant non-holonomic derivatives vanish such that it holds that $\bm{z}=\bm{y}$. Hence, the control vector fields expressed in privileged coordinates at the origin are
\begin{align}
    \bm{Z}_1 (\bm{z})\! &= \begin{bmatrix}
        \cos z_2 \cos \left( \frac{z_3}{\ell} \right) & 0 & -\sin z_2 & -\ell \cos z_2 \sin \left( \frac{z_3}{\ell} \right)
    \end{bmatrix}\tran, \nonumber \\ \bm{Z}_2 (\bm{z})\! &= \begin{bmatrix}
        0 & 1 & 0 & 0
    \end{bmatrix}\tran.
\end{align} The corresponding Taylor expansions are omitted for the sake of brevity. However, grouping all monomial vector fields of the expansion with relative degree $-1$, we obtain the homogeneous nilpotent approximation at the origin in terms of
\begin{align}\label{eq:approx_car}
    \dot{\bm{z}} = \bm{Z}_1^{[-1]} u_1 + \bm{Z}_2^{[-1]} u_2 = \begin{bmatrix}
        1 \\ 0 \\ -z_2 \\ - z_3
    \end{bmatrix} u_1 + \begin{bmatrix}
        0\\1\\0\\0
    \end{bmatrix} u_2,
\end{align} 
which is homogeneous with $\tau=0$, $\bm{s}=(1,\,1)$, and $\bm{r}=(1,\,1,\,2,\,3)$ and indeed approximates the kinematic car~\eqref{eq:kinematic_car} in the sense of~\cite[Def.\ 4.1]{CoronGrueneWorthmann20}. 
The very technical calculations to show the latter are skipped here for brevity. 
As previously mentioned, although the privileged coordinates at the origin can be obtained by reordering~$\bm{x}$, it nevertheless makes sense to apply Bella\"{i}che's algorithm since it yields a canonical form by taking into account the length~$\ell$, see Transformation~\eqref{eq:trafo_car}, potentially making easier controller tuning for different values of~$\ell$. 
Moreover, although not done here for reasons of brevity, choosing a setpoint with~$d_4\neq 0$ makes necessary the application of the full procedure since then~$\bm{z}\neq\bm{y}$. 
With the origin as setpoint, however, the tailored stage cost~\eqref{eq:MPC_cost} for the kinematic car is given in terms of
\begin{align} \label{eq:cost_car}
    \ell (\bm{x},\bm{u}) ={}&q_1 x_1^{12} + q_2 x_4^{12}  + q_3 (\ell x_3)^{6} + q_4 (\ell x_2)^4\nonumber\\ {}&+  r_1 v^{12} + r_2 \omega^{12}.
\end{align}
\begin{figure*}[t]
	\centering
        \begin{subfigure}[t]{0.64\linewidth}
		\centering
%
%
%

\begin{tikzpicture}

\begin{axis}[%
name=kinCarPlane,
width=3.2cm,
height=3.2cm,
at={(1cm,0)},
scale only axis,
xmin=-.1,
xmax=0.9,
ymin=-.4,
ymax=.6,
axis background/.style={fill=white},
xmajorgrids,
ymajorgrids,
xlabel={$x$ in $\unit{m}$},
ylabel={$y$ in $\unit{m}$}
] 

\addplot[color = colorPlot, line width = \linewidthPlots] table{src/data/KC_plane.txt}; %

\def\lengthDDc{0.2}
\def\heightDDc{0.15}
\readlist\KinCarX{0, 0.4715, 0.0001} 
\readlist\KinCarY{0.2, 0.0194, 0} 
\readlist\KinCarT{0, 7.1362, 0} 
\readlist\KinCarP{0, -6.2994, 0} 

\pic {KinCar={\KinCarX[1]}{\KinCarY[1]}{\KinCarT[1]}{\lengthDDc}{\heightDDc}{\KinCarP[1]}{fill=grayCar}{0.3}};
\node [above, shift={(axis direction cs:0.1,0.05)}, text opacity=0.3] at (axis cs:\KinCarX[1],\KinCarY[1]) {\small{$t=\unit[0]{s}$}};

\pic {KinCar={\KinCarX[2]}{\KinCarY[2]}{\KinCarT[2]}{\lengthDDc}{\heightDDc}{\KinCarP[2]}{fill=grayCar}{0.3}};
\node [below, shift={(axis direction cs:0.15,-0.05)}, rotate=7.1362, text opacity=0.3] at (axis cs:\KinCarX[2],\KinCarY[2]) {\small{$t=\unit[1.75]{s}$}};

\pic {KinCar={\KinCarX[3]}{\KinCarY[3]}{\KinCarT[3]}{\lengthDDc}{\heightDDc}{\KinCarP[3]}{fill=grayCar}{0.3}};
\node [below, shift={(axis direction cs:0.1,-0.05)}, rotate=0.6245, text opacity=0.3] at (axis cs:\KinCarX[3],\KinCarY[3]) {\small{$t=\unit[15]{s}$}};


\end{axis}

\begin{axis}[%
    width=6.4cm,
    height=3.2cm,
    at={(5.5cm, 0)},
    scale only axis,
    xmin=-1,
    xmax=1,
    ymin=-0.75,
    ymax=0.75,
    axis equal image,
    axis background/.style={fill=white},
    xmajorgrids,
    ymajorgrids,
    xlabel={$x$ in $\unit{m}$},
    ] 
    
    \def\distTr{0.2}        
    \def\lengthDDt{0.15}
    \def\heightDDt{0.1}
    \readlist\TrX{-0.4, 0.6291, -0.0113} 
    \readlist\TrY{0.2, -0.0876, -0.0001} 
    \readlist\TrTA{0, -1.2726, 0.2459} 
    \readlist\TrTB{0, -17.7822, -0.0293} 
    \readlist\TrTC{0, -9.4929, -0.0015} 
    
    \addplot[color = colorPlot2, line width = \linewidthPlots] table{src/data/T2_quadratic_plane.txt}; \label{plot:T2_quadratic}%
    
    \addplot[color = colorPlot, line width = \linewidthPlots] table{src/data/T2_plane.txt}; \label{plot:T2_tailored}%
    
    \pic {Trailer={\TrX[1]}{\TrY[1]}{\TrTA[1]}{\TrTB[1]}{\TrTC[1]}{\lengthDDt}{\heightDDt}{\distTr}{0.3}};
    \node [above, shift={(axis direction cs:-0.1,0.02)}, text opacity=0.3] at (axis cs:\TrX[1],\TrY[1]) {\small{$t=\unit[0]{s}$}};
    
    \pic {Trailer={\TrX[2]}{\TrY[2]}{\TrTA[2]}{\TrTB[2]}{\TrTC[2]}{\lengthDDt}{\heightDDt}{\distTr}{0.3}};
    \node [below, shift={(axis direction cs:0.08,-0.04)}, text opacity=0.3] at (axis cs:\TrX[2],\TrY[2]) {\small{$t=\unit[2.5]{s}$}};
    
    \pic {Trailer={\TrX[3]}{\TrY[3]}{\TrTA[3]}{\TrTB[3]}{\TrTC[3]}{\lengthDDt}{\heightDDt}{\distTr}{0.3}};
    \node [below, shift={(axis direction cs:-0.1,-0.04)}, text opacity=0.3] at (axis cs:\TrX[3],\TrY[3]) {\small{$t=\unit[25]{s}$}};
    
\end{axis}
\end{tikzpicture}%
		\caption[]%
		{{parallel parking of the kinematic car (left) and forward parking of the two-trailer vehicle (right)}}
		\label{fig:sim_plane}
	\end{subfigure}
	\hfill
    \begin{subfigure}[t]{0.33\linewidth}
		\centering
%
%
%
\begin{tikzpicture}

\begin{axis}[%
width=3.8cm,
height=1.3cm,
at={(0cm,1.9cm)},
scale only axis,
xmin=0,
xmax=15,
axis background/.style={fill=white},
xmajorgrids,
ymajorgrids,
ymode=log
] 

\addplot[color = colorPlot, line width = \linewidthPlots] table{src/data/KC_val_fct.txt}; %

\end{axis}

\begin{axis}[%
    width=3.8cm,
    height=1.3cm,
    at={(0cm, 0cm)},
    scale only axis,
    xmin=0,
    xmax=25,
    axis background/.style={fill=white},
    xmajorgrids,
    ymajorgrids,
    xlabel={$t$ in $\unit{s}$},
    ymode=log
    ] 
    
    \addplot[color = colorPlot, line width = \linewidthPlots] table{src/data/T2_val_fct.txt}; %

    \addplot[color = colorPlot2, line width = \linewidthPlots] table{src/data/T2_quadratic_val_fct.txt}; %
    
\end{axis}

\end{tikzpicture}%
		\caption[]%
		{{value function of the kinematic car (top) and the two-trailer system (bottom)}}
		\label{fig:sim_val_fct}
	\end{subfigure}
    \caption{Simulations of the kinematic car and the two-trailer system utilizing the proposed non-quadratic costs.}
    \label{fig:simulations}
\end{figure*}
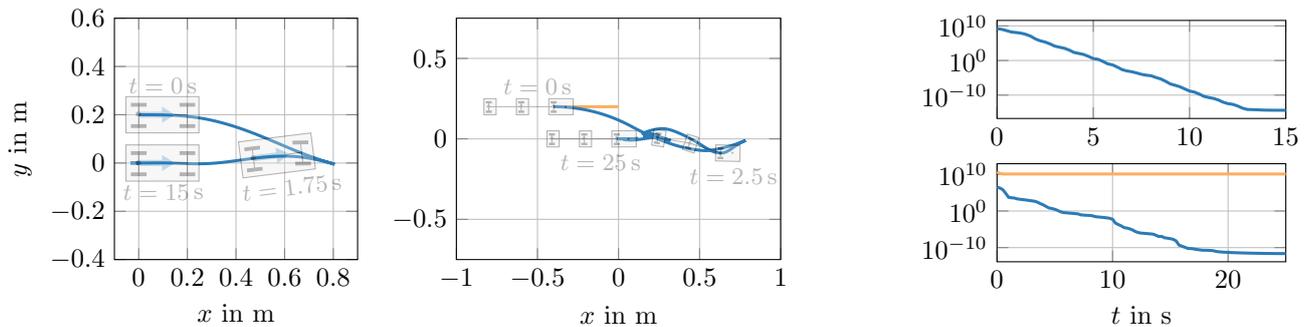%
In order to investigate the kinematic car's closed loop by means of simulations, the resulting OCP is formulated using CasADi~\cite{AnderssonEtAl19} through the Matlab interface. Note that this setup is utilized for all subsequent simulations and hardware experiments. 
The car's control bounds are set to $\mathcal{U} = [-1, \, 1]\ \unit{m/s} \times [-1, \,1 ]\ \unit{rad/s}$ and the distance between the two axles is chosen to $\ell=\unit[0.2]{m}$. 
Here and in the following, the controllers actually rely on the solution of a time-discretized version of the OCP~\eqref{eq:mpc_optprob_specific}, assuming a zero-order hold on the control input, a sampling time of $\delta_\mathrm{t}=\unit[0.25]{s}$, and a prediction horizon spanning $H=60$ discrete time steps. 
The reason to use a discrete-time controller of this type is that the robotic hardware employed in later experiments operates in such a fashion, with piecewise constant control inputs. 
For technical reasons, the OCP's cost function is scaled to obtain a more accurate solution. 
All OCPs in this paper are solved with IPOPT~\cite{Waechter2006}.
Since parking the car in parallel is particularly challenging, we consider a scenario where the kinematic car shall drive from its initial condition $\bm{x}_0=\begin{bmatrix} \unit[0]{m} & \unit[0.2]{m} & \unit[0]{rad} & \unit[0]{rad} \end{bmatrix}\tran$ to the origin, see also the quadratic cost's failure for this case in App.~\ref{app:proof}. 
The resulting  closed-loop trajectory in the plane and the corresponding value function are depicted in Figs.~\ref{fig:simulations}(\subref{fig:sim_plane}) and~\ref{fig:simulations}(\subref{fig:sim_val_fct}), respectively. 
Note that in Fig.~\ref{fig:simulations}(\subref{fig:sim_plane}), as well as later in Fig.~\ref{fig:experiments}, the spatial position of the (towing) vehicle's point $P$ in the plane is depicted over the course of time. 
Further, in Fig.~\ref{fig:simulations}(\subref{fig:sim_plane}), the vehicle's pose is depicted at selected time instants with reduced opacity in order to indicate its driving direction and configuration. 
It is evident for the kinematic car that, due to the higher exponent belonging to the deviation in the $x$-direction, the tailored cost function~\eqref{eq:cost_car} enables the car to maneuver far enough in order to approach the origin along the well-controllable direction~$\bm{X}_1$, see Fig.~\ref{fig:simulations}(\subref{fig:sim_plane}). 
Crucially, the remaining deviation after $\unit[15]{s}$ in the hardly-controllable directions $y$ and $\theta$ are less than $\unit[10^{-10}]{mm}$ and $\unit[ 10^{-4}]{^{\circ}}$, respectively. 
In contrast, when using standard quadratic costs but otherwise the same parameters and initial condition, the robot does not even begin to drive or converge to the setpoint but stays stationary in its initial pose. 

In view of the unicycle~\eqref{eq:unicycle} and the kinematic car~\eqref{eq:kinematic_car}, one could conjecture that a large class of non-holonomic vehicles is already expressed in privileged coordinates at the origin. However, this does not hold in general, as will be exemplified below.

\subsection{Differential drive with an attached trailer}
At first glance, attaching a trailer to a unicycle yields a system that is kinematically related to the kinematic car~\eqref{eq:kinematic_car}. 
However, differences arise when deriving a tailored~OCP. 
Using the state vector~$\bm{x}\tran = \begin{bmatrix}
    x & y & \theta & \theta_1
\end{bmatrix}$, the tandem's kinematics is given by
\begin{align}\label{eq:1trailer}
    \dot{\bm{x}} 
    = \begin{bmatrix}
        \cos \theta & 0 \\
        \sin \theta & 0 \\
        0 & 1\\
        \frac{1}{\ell_1} \sin (\theta-\theta_1) & 0
    \end{bmatrix} \begin{bmatrix}
        v \\ \omega
    \end{bmatrix} = \begin{bmatrix}
        \bm{X}_1 & \bm{X}_2
    \end{bmatrix} \bm{u},
\end{align} 
where the first three states describe the towing vehicle's pose analogously to Sec.~\ref{sec:unicycle}, and the fourth state describes the trailer's absolute orientation in the plane w.r.t. the $x$-axis of the inertial frame, see Fig.~\ref{fig:sketch_vehicles}. 
Since the overall vehicle is only actuated by the towing unicycle, the inputs are equivalent to those given in Sec.~\ref{sec:unicycle}. Kinematically, a second non-holonomic constraint arises restricting the velocity direction of the trailer. Nonetheless, system~\eqref{eq:1trailer} is controllable which can be shown by the LARC, based on the Lie brackets
\begin{align}
    \bm{X}_3  &= [\bm{X}_1, \, \bm{X}_2] = \begin{bmatrix}
        \sin \theta & -\cos \theta & 0 & - \frac{1}{\ell_1}\cos (\theta - \theta_1)
    \end{bmatrix}\tran,  \nonumber\\
    \bm{X}_4  &= [\bm{X}_1, \, \bm{X}_3] = \begin{bmatrix}
        0&0&0& -\frac{1}{\ell_1^2}
    \end{bmatrix}\tran.
\end{align} 
Then, the unicycle's distributions $\bm{\Delta}^1$ and $\bm{\Delta}^2$ as well as $\bm{\Delta}^3 = \text{span} \{\bm{X}_1, \, \bm{X}_2, \, \bm{X}_3, \, \bm{X}_4 \}=T_{\bm{0}}\mathbb{R}^4$ yield the growth vector $\bm{n}=(2,\,3,\,4)$, the weights $\bm{w}=(1,\,1,\,2,\,3)$, and, consequently, the degree of non-holonomy~$r=3$. 
Since, in the following, it is the goal to steer the system to its origin, i.e., $\bm{d}=\bm{0}$, the first-order non-holonomic derivatives of the coordinate functions at the origin are $\bm{X}_1 x_1 (\bm{0})=1$, $\bm{X}_1 x_j (\bm{0})=0,~j\in\mathbb{Z}_{2:4}$, $\bm{X}_2 x_3 (\bm{0})=1$, and $\bm{X}_2 x_k (\bm{0})=0,~k\in\mathbb{Z}_{1:4}\setminus\{ 3\}$. 
Further, non-vanishing second-order derivatives are $\bm{X}_2 \bm{X}_1 x_2 (\bm{0})=1$ and $\bm{X}_2 \bm{X}_1 x_4 (\bm{0}) = 1/\ell_1$ such that the system's non-holonomic orders at the origin follow as $\text{ord}_{\bm{0}}(x_1)=\text{ord}_{\bm{0}}(x_3)=1$ and $\text{ord}_{\bm{0}}(x_2)=\text{ord}_{\bm{0}}(x_4)=2$. 
Crucially, reordering the coordinates is not sufficient to ensure $\text{ord}_{\bm{0}}(x_j)=w_j$, implying that $\bm{x}$ is not a set of privileged coordinates. 
Note that this is in strong contrast to the unicycle and the kinematic car. In order to obtain privileged coordinates we deploy Bella\"{i}che's algorithm. From~\eqref{eq:Trafo1Bellaiche}, it follows at the origin that
\begin{align}
    \bm{y} = \begin{bmatrix}
        x_1 & x_3 & -x_2 & \ell_1 (x_2 - \ell_1 x_4)
        \label{eq:trafo_1trailer}
    \end{bmatrix}\tran
\end{align} and moreover, for the resulting vector fields $\bm{Y}_j,~j\in\mathbb{Z}_{1:3}$, the second step~\eqref{eq:Trafo2Bellaiche} vanishes at the origin such that it holds for the privileged coordinates that $\bm{z}=\bm{y}$ at $\bm{d}=\bm{0}$. 
For more general setpoints~$\bm{d}=\bm{0}$, which yields $\bm{z}\neq\bm{y}$, we refer to~\cite{Bellaiche1993}. 
For $\bm{d}=\bm{0}$, however, the control vector fields read
\begin{align}
    \bm{Z}_1 \! &= \! \begin{bmatrix}
        \cos z_2  & 0 & -\sin z_2 & \ell \sin z_2 - \ell \sin \left( z_2  + \frac{z_3}{\ell} + \frac{z_4}{\ell^2} \right)
    \end{bmatrix}\tran, \nonumber \\ 
    \bm{Z}_2 \! &= \! \begin{bmatrix}
        0 & 1 & 0 & 0
    \end{bmatrix}\tran,
\end{align} where we again omit listing their corresponding Taylor expansion. The resulting homogeneous nilpotent approximation follows from the monomial vector fields with weighted degree $-1$ and is the same as for the kinematic car~\eqref{eq:approx_car}. 
Again, the boundedness of the residuum of the difference of this approximation and the original kinematics~\eqref{eq:1trailer} can be shown as necessitated in~\cite{CoronGrueneWorthmann20}. This technical computation is skipped for the sake of brevity. 
Thus, the tailored stage cost for the unicycle with an attached trailer follows from the homogeneity of~\eqref{eq:approx_car} and the transformation of the tandem into the privileged coordinates $\bm{z}$, yielding
\begin{align}\label{eq:cost_1trailer}
    \ell (\bm{x},\bm{u}) ={}&q_1 x_1^{12} + q_2 x_3^{12} \nonumber + q_3 (x_2)^{6} + q_4(\ell_1 (x_2 - \ell_1 x_4) )^4 \\&+  r_1 v^{12} + r_2 \omega^{12}.
\end{align}
It is worth noting that the (scaled) fourth privileged coordinate $z_4/\ell_1 = x_2 - \ell_1 x_4$ is the small-angle approximation of the deviation of the trailer's axle center in the $y$-direction, i.e., in the direction hardest to control. 
This intrinsic property is taken into account in the OPC due to the smaller exponent, which implies that the controller prioritizes the minimization of this error close to the origin. 
Consequently, the MPC first aims to steer the center point of the trailer's axle to the $x$-axis. 
In the next phase, it is the goal to reduce the deviation of the towing unicycle in the $y$-direction and, finally, the origin is approached along the well controllable $x$-axis. 
This characteristic closed-loop behavior can also be observed in hardware experiments.

\begin{figure}[t]
	\centering
    \begin{subfigure}[t]{0.49\linewidth}
		\centering
		\includegraphics[width=0.95\textwidth]{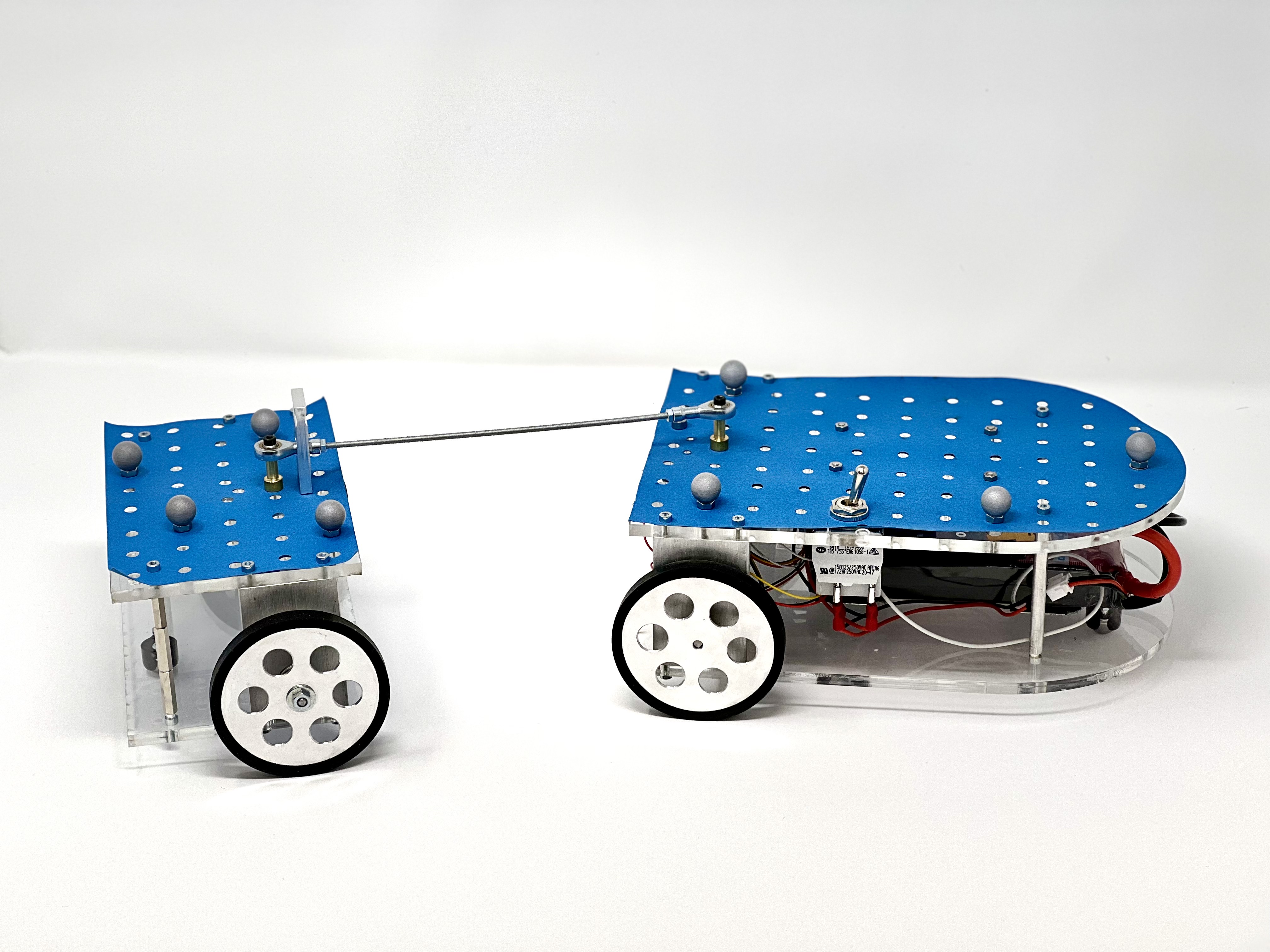}
	\end{subfigure}
	\hfill
    \begin{subfigure}[t]{0.49\linewidth} 
		\centering
        \def\clipR{0}
        \def\clipL{0}
        \def\clipB{0}
        \def\clipT{0}
        \def\SVGwidthC{0.95\textwidth}
		\def\svgwidth{\SVGwidthC}
\begingroup%
  \makeatletter%
  \providecommand\color[2][]{%
    \errmessage{(Inkscape) Color is used for the text in Inkscape, but the package 'color.sty' is not loaded}%
    \renewcommand\color[2][]{}%
  }%
  \providecommand\transparent[1]{%
    \errmessage{(Inkscape) Transparency is used (non-zero) for the text in Inkscape, but the package 'transparent.sty' is not loaded}%
    \renewcommand\transparent[1]{}%
  }%
  \providecommand\rotatebox[2]{#2}%
  \newcommand*\fsize{\dimexpr\f@size pt\relax}%
  \newcommand*\lineheight[1]{\fontsize{\fsize}{#1\fsize}\selectfont}%
  \ifx\svgwidth\undefined%
    \setlength{\unitlength}{225.39095306bp}%
    \ifx\svgscale\undefined%
      \relax%
    \else%
      \setlength{\unitlength}{\unitlength * \real{\svgscale}}%
    \fi%
  \else%
    \setlength{\unitlength}{\svgwidth}%
  \fi%
  \global\let\svgwidth\undefined%
  \global\let\svgscale\undefined%
  \makeatother%
  \begin{picture}(1,0.74926796)%
    \lineheight{1}%
    \setlength\tabcolsep{0pt}%
    \put(0,0){\includegraphics[width=\unitlength,page=1]{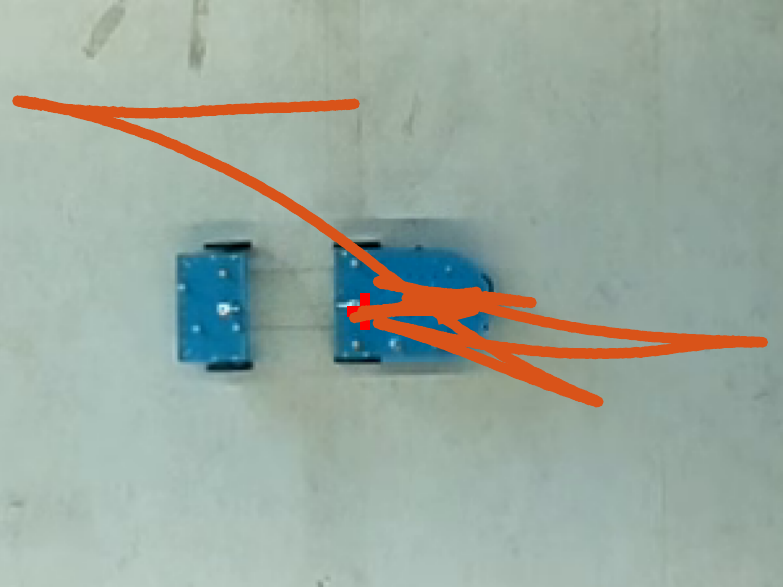}}%
    \put(0.0118159,0.0114315){\makebox(0,0)[lt]{\lineheight{1.25}\smash{\begin{tabular}[t]{l}$t = 19.00\,\textnormal{s}$\end{tabular}}}}%
  \end{picture}%
\endgroup%

	\end{subfigure}
    \caption{Photograph of the experimental one-trailer system (left) and corresponding parallel parking scenario (right).}
    \label{fig:experiments}
\end{figure}%

The same setup as previously described for the kinematic car is utilized, but the unicycle's control bounds are tightened to $\mathcal{U} = [-0.4, \, 0.4]\ \unit{m/s} \times [-\pi/8, \, \pi/8]\ \unit{rad/s}$ due to the employed electric motors. 
While the corresponding MPC runs at a sampling time of $\unit[250]{ms}$ on a standard laptop, to govern the intended velocity following from the OCP~\eqref{eq:mpc_optprob_specific}, a nested PID controller is running at a frequency of $\unit[100]{Hz}$ on a BeagleBone Blue board~\cite{BeagleBoard.orgFoundation22}. 
The communication between the laptop and the on-board computer is realized in the form of multicast messages using the LCM communication library~\cite{HuangOlsonMoore10}. 
Further, the vehicle's pose and configuration is provided by six Optitrack Prime 13W ultra-wide angle cameras mounted in the laboratory environment.

In the considered scenario, the one-trailer vehicle with $\ell_1\coloneqq\unit[0.19]{m}$, see Fig.~\ref{fig:experiments} (left), shall perform a parallel parking maneuver starting from the initial pose
$\bm{x}(0)=\begin{bmatrix} \unit[0]{m} & \unit[0.3]{m} & \unit[0]{rad} & \unit[0]{rad} \end{bmatrix}\tran$. 
In the hardware experiment, initial placement cannot be exact, but deviations have absolute values smaller than $\unit[3]{mm}$ and $\unit[0.035]{rad}$ in the positions and orientations, respectively. 
The goal is to approach the origin until a prescribed, desired accuracy is achieved. 
Crucially, the proposed non-quadratic controller succeeds by significantly reversing the vehicle in order to be able to compensate the deviations in the challenging control directions, i.e., those following from nested control actions and with weights greater than one.
For the scenario depicted in Fig.~\ref{fig:experiments}, the controller is terminated after about $\unit[20]{s}$ with an absolute deviation to the origin of less than~$\begin{bmatrix} \unit[0.2]{mm} & \unit[0.5]{mm} & 2^{\circ} & 0.5^{\circ}  \end{bmatrix}\tran$. 
Although this accuracy may not always be achieved in each attempt after the aforementioned time span, it is crucial to note that the proposed non-quadratic cost yields a controller which continuously seeks to minimize the cost function as long as the vehicle is not exactly parked at the origin.
In contrast to this, the MPC controller using the standard quadratic cost does not attempt to reduce the distance to the origin at all times but rather stays practically stationary at a point that is not the origin.
For the considered scenario, the controller yields a significant deviation in the fourth state of $34^{\circ}$.
Once again, the particular quantitative result might be strongly influenced by the chosen weights and prediction horizon. 
Still, the previous scenario stresses the theoretically proven structural problem of the quadratic controller penalizing the deviation w.r.t.\ Euclidean geometry.
Whilst in conventional holonomic mechanical systems, the steady-state offset usually is accounted for by plant-model mismatches, disturbances, or the actuators' sensitivity, in the present case, it is an intrinsic problem of the quadratic costs regarding the non-holonomic vehicles. 
Notice that videos of the scenarios discussed previously can be found in the digital supplementary material provided with this article.

Comparing the obtained stage cost~\eqref{eq:cost_1trailer} with the unicycle's cost~\eqref{eq:cost_unicycle} indicates that the order of the obtained stage cost scales with the kinematic complexity of the concerned system~\eqref{eq:system}, primarily w.r.t. its degree of non-holonomy. This is elucidated by adding another trailer to the vehicle.

\subsection{Differential drive with two attached trailers}
The kinematics of the system is an extension of the kinematics of the one-trailer system~\eqref{eq:1trailer} with an additional fifth state describing the absolute orientation of the second trailer w.r.t. the positive $x$-axis, see Fig.~\ref{fig:sketch_vehicles}(\subref{fig:sketch_trailer2}). The additional kinematics can be described by means of $\dot{\theta}_2 = \dot{x}_5 = \cos(\theta -\theta_1)\sin(\theta_1 - \theta_2)/\ell_2$ such that the control vector fields are
\begin{align}
    \bm{X}_1 &= \begin{bmatrix}
        \cos \theta & \sin \theta & 0 & \frac{\sin (\theta - \theta_1)}{\ell_1} & \frac{\cos(\theta-\theta_1) \sin(\theta_1-\theta_2)}{\ell_2}
    \end{bmatrix}\tran, \nonumber\\ 
    \bm{X}_2 &= \begin{bmatrix}
        0 & 0 & 1 & 0 & 0
    \end{bmatrix}\tran. \label{eq:2trailer}
\end{align} 
For more details regarding the kinematics of this system or the general $N$ trailer case, the interested reader is referred to~\cite{Sordalen1993}. 
Besides the two non-holonomic constraints arising for the one-trailer system, the second trailer is subject to an analogous non-holonomic constraint. However, using the Lie brackets $\bm{X}_3 = [\bm{X}_1, \, \bm{X}_2]$, $\bm{X}_4 = [\bm{X}_1, \, \bm{X}_3]$, and $\bm{X}_5 = [\bm{X}_1, \, \bm{X}_4]$, the system's controllability can be shown, e.g., via the LARC. 
Then, choosing the distributions $\bm{\Delta}^j,~j\in\mathbb{Z}_{1:4}$, in the same fashion as before, the growth vector $\bm{n}=(2,\,3,\,4,\,5)$ and the weights $\bm{w}=(1,\,1,\,2,\,3,\,4)$ are obtained. Again, it is the goal to asymptotically stabilize the origin, i.e., it holds subsequently that $\bm{d}=\bm{0}$, and the corresponding non-vanishing first-order non-holonomic derivatives read $\bm{X}_1 x_1 (\bm{0})=1$ and $\bm{X}_2 x_3 (\bm{0})=1$. Further relevant higher-order derivatives at the origin are $\bm{X}_2 \bm{X}_1 x_2 (\bm{0})=1$, $\bm{X}_2 \bm{X}_1 x_4 (\bm{0})=1/\ell_1$, and $\bm{X}_2 \bm{X}_1^2 x_5 (\bm{0})=1/(\ell_1 \ell_2)$ such that it holds that $\text{ord}_{\bm{0}}(x_j)=1,~j\in\{1,3\}$, $\text{ord}_{\bm{0}}(x_k)=2,~k\in\{2,4\}$, and $\text{ord}_{\bm{0}}(x_5)=3$.
Evidently, as for the single-trailer system, reordering $\bm{x}$ is not sufficient to obtain a set of privileged coordinates and thus, we employ Bella\"{i}che's algorithm. 
Concerning the origin, applying the first transformation step~\eqref{eq:Trafo1Bellaiche} yields $y_1 = x_1$, $y_2 = x_3$, $y_3 = -x_2$, $y_4 = x_2(\ell_1 + \ell_2) - x_4 \ell_1 (\ell_1+\ell_2) - x_5 \ell_2^2$, and $y_5 = -x_2 \ell_1 \ell_2 + x_4 \ell_1^2 \ell_2 + x_5 \ell_1 \ell_2^2$. 
Once again, the second step~\eqref{eq:Trafo2Bellaiche} does not imply any changes at the origin, i.e., it holds that $\bm{z}=\bm{y}$. 
This example makes clear that an increasing degree of non-holonomy renders the derivation of the privileged coordinates, and thus, of the resulting MPC cost, increasingly more sophisticated. However, since the proposed procedure is constructive, the crucial ingredients can still be derived in a straightforward fashion, with merely the calculations and obtained expressions becoming lengthier. 

Therefore, the system's kinematics can be expressed in terms of the privileged coordinates~$\bm{z}$ and in the next step, the corresponding Taylor expansion can be determined. Then, grouping all monomial vector fields of the expansion with weighted degree $-1$, the homogeneous nilpotent approximation for the unicycle with two attached trailers reads
\begin{align}\label{eq:approx_2trailer}
    \dot{\bm{z}} = \bm{Z}_1^{[-1]} u_1 + \bm{Z}_2^{[-1]} u_2 = \begin{bmatrix}
        1 \\ 0 \\ -z_2 \\ - z_3 \\ -z_4
    \end{bmatrix} u_1 + \begin{bmatrix}
        0\\1\\0\\0\\ 0
    \end{bmatrix} u_2
\end{align}
which is homogeneous with $\tau=0$, $\bm{s}=(1,\,1)$, and $\bm{r}=(1,\,1,\,2,\,3, \, 4)$ and approximates the original system with the residuum fulfilling the conditions given in~\cite[Def.\ 4.1]{CoronGrueneWorthmann20}, although, again, we refrain from printing the residuum calculations for brevity. 
Consequently, the tailored MPC stage cost for the system approximated at the origin follows from~\eqref{eq:MPC_cost} and the state transformation, yielding
\begin{align}
    \ell (\bm{x},\bm{u}) ={}&q_1 x_1^{48} + q_2 x_3^{48}  + q_3 x_2^{24} \nonumber \\
    &+ ( x_2(\ell_1 + \ell_2) - x_4 \ell_1 (\ell_1+\ell_2) - x_5 \ell_2^2)^{16} \nonumber \\
    &+ (-x_2 \ell_1 \ell_2 + x_4 \ell_1^2 \ell_2 + x_5 \ell_1 \ell_2^2)^{12} \nonumber \\
    &+  r_1 v^{48} + r_2 \omega ^{48}.  \label{eq:cost_2trailer} 
\end{align}
As previously mentioned, the obtained cost function becomes more intricate the more restricted the vehicle's movement capabilities are, which is reflected by the degree of non-holonomy $r = w_{n_r}$. 
Hence, numerically solving the OCP becomes more challenging.
However, the non-quadratic stage cost~\eqref{eq:cost_2trailer} implicitly takes into account the non-holonomic nature of the vehicle such that its origin is locally asymptotically stabilized. 
This can also be observed by means of the following simulation scenario. 
Therein, the vehicle with $\ell_1=\ell_2=\unit[0.2]{m}$ shall drive from its initial pose $\bm{x}_0=\begin{bmatrix} \unit[-0.4]{m} &\unit[0.2]{m} &\unit[0]{rad}&\unit[0]{rad}&\unit[0]{rad}\end{bmatrix}\tran$ to the origin. 
Due to the more intricate kinematics of the two-trailer system, the prediction horizon is increased to $H=80$ while the other aforementioned parameters read the same as for the kinematic car.
Then, as indicated by the closed-loop trajectory~\eqref{plot:T2_tailored} in Fig.~\ref{fig:simulations}(\subref{fig:sim_plane}), the proposed controller successfully steers the system into the origin. 
Moreover, the corresponding value function strictly decreases until a certain value is met, see Fig.~\ref{fig:simulations}(\subref{fig:sim_val_fct}).
Note that the tailored cost function is again scaled for technical purposes such that its quantitative values are of less importance compared to its qualitative behavior. 
Crucially, the remaining deviations in the hardly controllable directions $y$, $\theta_1$, and $\theta_2$ after $\unit[15]{s}$ are less than $\unit[0.2]{mm}$, $3\cdot 10^{-2}{\,}^{\circ}$, and $2\cdot 10^{-3}{\,}^{\circ}$, respectively.
In contrast to this, the quadratic cost's failure is shown by means of its corresponding closed-loop trajectory~\eqref{plot:T2_quadratic} in Fig.~\ref{fig:simulations}(\subref{fig:sim_plane}) and the related, stagnating value function in Fig.~\ref{fig:simulations}(\subref{fig:sim_val_fct}). 

We conjecture that the solution of the OCPs of the form~\eqref{eq:mpc_optprob_specific}, which need to be solved for the proposed controllers, becomes more difficult for larger exponents. 
The latter is connected to an increasing degree of non-holonomy. 
That is why, for the kinematic car and the considered unicycle-trailer combinations, cost functions with cancelled exponents are implemented, cf.~\cite[Rem.~4.7]{CoronGrueneWorthmann20}. 
In particular, all exponents have been divided by their greatest common divisor, yielding for example $\ell(\bm{z},\bm{u})= \sum_{i=1}^5 q_i \left| z_i \right|^{e_j} + \sum_{j=1}^2 r_j \left| u_j \right|^{f_j}$, where $\bm{e}=(12, \, 12, \, 6, \, 4,\, 3)$ and $\bm{f}=(12, \,12)$, for the two-trailer vehicle.
Although this choice yields the smallest integer exponents, choosing the cost such that all exponents are even may be worthwhile to investigate a more efficient solution by leveraging sum-of-squares-techniques. 
However, as the favorable results show, for the vehicles considered in this paper, the numerical performance was sufficient without further measures and we leave more profound investigations on numerical solution techniques for OCP~\eqref{eq:mpc_optprob_specific} to future work.

\section{Conclusions}
This paper's proposed design process for nonlinear model predictive controllers for driftless, controllable non-holonomic vehicles is comprehensive and results in a closed loop that is provably nominally asymptotically stable if the prediction horizon is sufficiently long. 
As Sec.~\ref{sec:quadratic_cost} shows, this is not given automatically and naturally since the usual and almost universally used MPC design paradigm involving a quadratic cost function provably must fail for all example vehicles considered in this paper. 
Moreover, this paper's application examples show that the proposed design procedure, while being far from trivial, is indeed practically feasible for the application to vehicles with great real-world significance. 
Hence, we believe that this paper can serve control engineers as a blueprint for the design of controllers for their kinds of non-holonomic vehicles. 
The latter may not only include ground vehicles, but also underwater and aerial vehicles. 
Further possible applications are planar space robots or a kinematic hopping robot in its flight phase, see, e.g.,~\cite{Murray2017}.

Apart from this immediate usefulness of the paper's findings, there are further features and peculiarities of the proposed approach that warrant a closer inspection in future work. 
Firstly, the system approximation procedure used in the overall MPC design procedure results in a canonical form~\cite{Bellaiche1993} that closely resembles the so-called chained form, see the Approximations~\eqref{eq:approx_car} and~\eqref{eq:approx_2trailer}, for instance. 
The chained form has over decades been of paramount importance regarding non-holonomic systems, e.g., for motion planning~\cite{Murray1991}. 
Hence, obtaining such a canonical form could facilitate the tuning of weighting factors as well as choosing a sufficiently long prediction horizon in future work.  
Secondly, results from our previous work~\cite{RosenfelderEbelEberhard21} seem to give an indication that MPC controllers with a cost function designed according to this paper's design principle may also work well if non-holonomic vehicles are not modelled kinematically but in the form of a second-order model incorporating inertia effects. 
Future work may try to formalize this in a similar way as done in this paper for the first-order case. 
Thirdly, the design procedure is not limited to the control of individual vehicles but may be applied to a system composed of multiple, potentially different, non-holonomic vehicles. 
Solving the resulting optimal control problem in a distributed fashion may lead to a distributed controller constructively tailored to the overall system's non-holonomic constraints, potentially exceeding the performance observed in our own previous undertakings in that direction~\cite{RosenfelderEbelEberhard22}. 
Hence, the latter is subject of our ongoing research work, and we intend to use such a controller in physically meaningful benchmark examples from robotics, e.g., as in~\cite{Ebel21}.

\section*{Acknowledgements}
This work was supported by the Deutsche Forschungsgemeinschaft (DFG, German Research Foundation) under Grant 433183605 and through Germany's Excellence Strategy (Project PN4-4 Theoretical Guarantees for Predictive Control in Adaptive Multi-Agent Scenarios) under Grant EXC 2075-390740016 (SimTech). 

\bibliographystyle{plain}        
\bibliography{literature_Ebel_22b_abbrev}

\appendix
\section{Insufficiency of quadratic costs for considered non-holonomic vehicles} \label{app:proof}   
In order to prove that choosing quadratic costs for the considered vehicles does not necessarily yield an asymptotically stable closed loop, we show that, in each neighborhood of the origin, there exists for each vehicle an initial state satisfying Condition~\eqref{eq:lemma_1}. The following arguments are an extension of the proof showing that this condition holds for the unicycle~\cite[Lem.~11]{MuellerWorthmann17}. Thus, technicalities are kept brief with reference to~\cite{MuellerWorthmann17}. Note that Condition~\cite[Eq.~(13)]{MuellerWorthmann17} is equivalent to Condition~\eqref{eq:lemma_1} for the unicycle~\eqref{eq:unicycle} reading
\begin{subequations}\label{eq:DIANA_condition}
    \begin{align}
        \cos \theta_0 \left( q_{11} x_0 + q_{12} y_0 + q_{13} \theta_0\right) \quad& \nonumber \\ + \sin \theta_0 \left( q_{12} x_0 + q_{22} y_0 + q_{23} \theta_0\right) &= 0, \\
        q_{13} x_0 + q_{23} y_0 + q_{33} \theta_0 &= 0, \label{eq:DIANA_condition_b}
    \end{align}     
\end{subequations}
where $q_{ij}$ denotes the corresponding entry of $\bm{Q}=\bm{Q}\tran$. Hence, we proceed with investigating the kinematic car.
\subsection{Kinematic car}
For the front-wheel driven kinematic car~\eqref{eq:kinematic_car}, 
Condition~\eqref{eq:lemma_1} degenerates for the choice $\phi_0 = 0$ to~\eqref{eq:DIANA_condition} with the marginal difference that~\eqref{eq:DIANA_condition_b} contains the entries $q_{i4}$ instead of $q_{i3},~i\in\mathbb{Z}_{1:3}$. 
However, in spite of that, a case differentiation very analogous to that of the unicycle can be employed, see~\cite{MuellerWorthmann17}, but with case conditions involving~$q_{i4}$ instead of $q_{i3}$. 
Note that the same set of equations is obtained for the rear-wheel driven kinematic car if the initial steering angle is chosen to zero~\cite{RosenfelderEbelEberhard21}. 

\subsection{Unicycle with an attached trailer}
Condition~\eqref{eq:lemma_1} can be written for this system~\eqref{eq:1trailer} as 
\begin{subequations}\label{eq:1trailer_condition}
    \begin{align}
        \cos \theta_0 \left( q_{11} x_0 + q_{12} y_0 + q_{13} \theta_0\right) \quad& \nonumber \\ + \sin \theta_0 \left( \tilde{q}_{12} x_0 + \tilde{q}_{22} y_0 + \tilde{q}_{23} \theta_0\right) &= 0, \\
        q_{13} x_0 + q_{23} y_0 + q_{33} \theta_0 &= 0, \label{eq:1trailer_condition_b}
    \end{align}     
\end{subequations} where $\tilde{q}_{12}\coloneqq q_{12}+q_{41}$, $\tilde{q}_{22}\coloneqq q_{22}+q_{42}$, and $\tilde{q}_{23} \coloneqq q_{23}+q_{43}$. Therein, we have set~$\theta_{1,0}=\unit[0]{rad}$ and, purely for the sake of readability but without loss of generality, $\ell_1=\unit[1]{m}$. 
The existence of an initial state $\bm{x}_0$ that is arbitrarily close to the origin and that satisfies~\eqref{eq:1trailer_condition} can again be concluded using the case distinctions given in~\cite{MuellerWorthmann17}, due to the analogy of~\eqref{eq:1trailer_condition} and~\eqref{eq:DIANA_condition}. 
As shown below, this can further be extended to the two-trailer system.
\subsection{Unicycle with two attached trailers}
Setting $\theta_{j,0}=\unit[0]{rad}$ and, again without loss of generality, $\ell_j=\unit[1]{m}$, $j\in\mathbb{Z}_{1:2}$, Condition~\eqref{eq:lemma_1} reads for the two-trailer system's kinematics, see the control vector fields~\eqref{eq:2trailer}, exactly the same as for the one-trailer case~\eqref{eq:1trailer_condition} such that no further proof is needed. 
Remarkably, the previous ansatz of setting $\theta_{j,0}=\unit[0]{rad}$ for all~$j$ can directly be extended to an $N$-trailer system~\cite{Sordalen1993} showing the necessity of non-quadratic stage costs for towing vehicles of this form.

\textbf{Remark}
Since it is sufficient to show that there exists one initial state in an arbitrarily small neighborhood of the origin satisfying Condition~\eqref{eq:lemma_1}, we have restricted ourselves to a natural extension of~\cite{MuellerWorthmann17}, i.e., choosing the additional initial orientations as zero. However, it seems natural that for an increasing degree of non-holonomy, there exists a growing set of initial states leading to the quadratic cost's failure. 

\end{document}